\documentclass[a4paper]{article}
\usepackage{graphicx}
\usepackage{color}
\usepackage[final,mode=multiuser]{fixme}
\usepackage{rotating}
\usepackage{comment}
\usepackage{multirow}
\usepackage{latexsym}
\usepackage[ruled,vlined,linesnumbered]{algorithm2e}
\usepackage{amsmath}
\usepackage{amssymb}
\usepackage{url}
\usepackage{cite}
\usepackage{tabularx}

\fxuseenvlayout{colorsig}
\fxusetargetlayout{color}
\FXRegisterAuthor{tt}{att}{TT}
\FXRegisterAuthor{si}{asi}{SI}
\FXRegisterAuthor{ha}{aha}{HA}
\fxsetface{env}{}
\fxsetface{inline}{\bfseries}

\renewcommand{\subsubsection}[1]{\noindent \textbf{#1}}

\newcommand{\argmin}{\mathop{\rm arg~min}\limits}
\newcommand{\set}[1]{\{ #1 \}}

\newtheorem{theorem}{Theorem}
\newtheorem{lemma}{Lemma}

\newtheorem{definition}{Definition}

\newtheorem{proof}{Proof}

\newcommand{\qed}{\nobreak \ifvmode \relax \else
      \ifdim\lastskip<1.5em \hskip-\lastskip
      \hskip1.5em plus0em minus0.5em \fi \nobreak
      \vrule height0.75em width0.5em depth0.25em\fi}

\begin{document}

\title{Efficient Constrained Pattern Mining Using Dynamic Item Ordering for Explainable Classification}

\author{Hiroaki Iwashita$^1$ \quad Takuya Takagi$^1$ \quad Hirofumi Suzuki$^1$ \\
Keisuke Goto$^1$  \quad Kotaro Ohori$^1$ \quad Hiroki Arimura$^2$ \\
{$^1$ Fujitsu Laboratories Ltd., Japan}\\
\texttt{\{iwashita.hiroak, takagi.takuya, suzuki-hirofumi, } \\
\texttt{goto.keisuke,  ohori.kotaro\}@fujitsu.com}\\
{$^2$ Graduate School of IST, Hokkaido University, Japan}\\
{\texttt{arim@ist.hokudai.ac.jp}}\\}

\date{}

\maketitle

\pagestyle{plain}

\begin{abstract}
Learning of interpretable classification models has been attracting much attention for the last few years. Discovery of succinct and contrasting patterns that can highlight the differences between two classes are very important. Such patterns are useful for human experts, and can be used to construct powerful classifiers. In this paper, we consider mining of minimal emerging patterns from high-dimensional data sets under a variety of constraints in supervised setting. We focus on an extension in which patterns can contain negative items that designate the absence of an item. In such a case, a database becomes highly dense, and it makes mining more challenging since popular pattern mining techniques such as fp-tree and occurrence deliver do not efficiently work. To cope with this difficulty, we present an efficient algorithm for mining minimal emerging patterns by combining two techniques: dynamic variable-ordering during pattern search for enhancing pruning effect, and the use of a pointer-based dynamic data structure, called dancing links, for efficiently maintaining occurrence lists. Experiments on benchmark data sets showed that our algorithm achieves significant speed-ups over emerging pattern mining approach based on LCM, a very fast depth-first frequent itemset miner using static variable-ordering. 
\end{abstract}


\section{Introduction}
\label{sec:intro}
  
Machine learning of various classes of interpretable prediction models over combinatorial features, such as decision trees and rule lists~\cite{Angelino_JMLR2018,Lakkaraju_KDD2016}, attracts much attention for last a few years from the view of trustable machine learning and knowledge discovery. 
Among many classes of combinatorial features, \textit{constrained patterns} such as \textit{contrast} and \textit{emerging patterns} are important classes of combinatorial features in high-dimensional data sets~\cite{Dong_KDD1999,Fan_TKDE2006,Loekito_KDD2006}, which are itemsets that discriminate one class from another by capturing significant differences among two classes.  
%
%
These classes of patterns are useful to capture high difference in two data sets, to provide human experts interpretable explanation, and to construct highly accurate classifiers~\cite{Li_PAKDD2000}.

Techniques in modern frequent itemset miners, such as LCM~\cite{Uno_FIMI2004}, work also well in finding constrained patterns of a sparse dataset where the frequency drops sharply with the addition of items.
However, for knowledge discovery, we often work with dense databases.
For example, we consider the case that a pattern consists of positive as well as negative items, where a negative item is a special symbol $\bar i$ indicating that the corresponding  positive item $i$ does not appear in a transaction data. 
This is important for interpretability and knowledge discovery because it allows us to describe patterns with fewer combinations of features that would be difficult to express succinctly with only positive items. 
%
%
%
%
%

We propose a mining algorithm for constrained patterns that efficiently works on not only sparse databases and also dense databases.
The key technique of our algorithm is to apply dynamic item ordering during pattern search.
In our algorithm, we use several pruning methods based on dynamic item ordering, and some of which are very effective for dense databases.
The same idea is used for maximal frequent pattern mining~\cite{Bayardo_SIGMOD1998}, but as the best of our knowledge, it has not been considered for constrained pattern mining.
In order to efficiently work dynamic item ordering, we also propose a novel representation of database \textit{DRMX} (Dynamically Reducible Binary Matrix) based on dancing links~\cite{Knuth_2000} which supports the deletion of rows and columns in the arbitrary order at any moment, and undo them in the reverse order to restore the previous snapshot.

By experiments on real data sets, we compare our mining algorithms {MiningMCP} with the previous, state-of-the-art algorithms in both mining and learning tasks. 
After confirming the effectiveness of dynamic ordering in various pruning strategies, we compare the proposed method {MiningMCP} with the state-of-the-art methods LCM~\cite{Uno_FIMI2004} and CP-tree~\cite{Fan_TKDE2006} for mining jumping emerging patterns. 
We observed that {MiningMCP} is 100 to 1000 times faster than LCM and CP-trees for almost all dense data sets.
Finally, we conducted binary classification experiments, and observed that the models constructed by our method achieved superior accuracy in all data sets than existing learning methods such as logistic regression, decision tree
ls, and random forests and that the use of negative items was effective in learning some difficult data sets. 

This paper is organized as follows. Section~\ref{sec:pre} gives the preliminaries. The details of our method is provided in section~\ref{sec:alg}. Section~\ref{sec:exp} presents  experimental results. Section~\ref{sec:conc} is conclusion of our paper.

\section{Preliminaries}
\label{sec:pre}

\subsection{Labeled databases and generalized itemsets}
Let $I = \{a_1, \dots, a_n\}$ be an alphabet of $n$ items. 
A \textbf{labeled database} over $I$ is a pair $D = (D_+, D_-)$, where $D_+, D_- \subseteq 2^I$ are possibly overlapping sets of positive and negative tuples over $I$, respectively. 
A tuple in $D$ is also called a data or an example. 
As a class of patterns, we consider the class of generalized itemsets defined as follows. 
A \textbf{literal} is either an item $x \in I$ or its negation $\neg x$. We refer to $x$ and $\neg x$ as \textbf{positive} and \textbf{negative literals}. We denote the \textbf{set of all negative literals} by $\neg I := \{\neg x \mid x \in I\}$.
We denote by $D = D_I := 2^{I\cup\neg I}$ the domain of all possible labeled databases over $I$.  

A \textbf{generalized itemset} (a pattern, for short) over $I$ is an expression $X = X_{pos} \cup X_{neg}$, where $X_{pos} = \{x_1, \ldots, x_k\} \subseteq I$ and $X_{neg} = \{\neg x_{k+1}, \ldots, \neg x_{k+m}\} \subseteq \neg I$ are sets of $k$ positive literals and $m$ negative literals, respectively. 
Then, the \textbf{size} of $X$ is $|X| = k  + m$. Clearly, $X \subseteq I \cup \neg I$. In what follows, we denote by $\mathbb P = 2^{I \cup \neg I}$ the class of  generalized itemsets over $I$. 
For any tuple $t \in 2^I$, a generalized itemset $X$ \textbf{occurs in} $t$, denoted $X \sqsubseteq t$, if all positive literals and none of negative literals of $X$ are contained in $t$, i.e. $\forall i \in [1..k], x_i \in t$ and $\forall j \in [k+1..k+m], x_j \notin t$. For any tuple $t \in D_+\cup D_-$, if $X \sqsubseteq t$, we say that $t$ is an occurrence of $X$ in $D$. For any set $D$ of tuples, the  occurrence list  of $X$ in $D$ is the set $Occ_{D}(X) := \{ t \in D \mid X \sqsubseteq t \}$. The positive and negative  supports are $Sup_+(X) := |Occ_{D_+}(X)|$ and $Sup_-(X) := |Occ_{D_-}(X)|$, respectively.
In terms of propositional logic,  a generalized itemset $X = \{x_1, \ldots, x_k\} \cup \{\neg x_{k+1}, \ldots, \neg x_{k+m}\}$ represents the conjunction $$\widetilde X := (\bigwedge _{i=1}^k x_i) \wedge (\bigwedge _{j=k+1}^{k+m} \neg x_j)$$ of positive and negative literals over $I$. The logical meaning of $X$ is given as follows. For any assignment $t \in 2^I$, we define the associated Boolean assignment $\tilde t: I \to \{0,1\}$ as $\tilde t(x) = 1$ if $x \in t$ and $\tilde t(x) = 0$ otherwise.  Then, we can easily show that $X \sqsubseteq t$ if and only if the conjunction $\widetilde X$ is valid on $\tilde t$, that is, $\tilde t \models \widetilde X$.

\subsection{Our data mining problem}


Let $\mathbb{LD}$ and $\mathbb P$ be domains of labeled databases and  patterns over $I$. A \textbf{pattern constraint} (or constraint)
over $\mathbb{LD}$ and $\mathbb P$ is a mapping 
$\mathbb C(\cdot \mid \mathbb P): \mathbb{LD} \to \mathbb P$  
that assigns a given labeled database $D = (D_+, D_-) \in \mathbb{LD}$ to a subset $\mathbb C(D \mid \mathbb P)\subseteq \mathbb P$ of patterns. 
We will simply refer to $\mathbb C = \mathbb C(D\mid \mathbb P)$ as a \textbf{constraint} on $\mathbb P$ if $D$ is clear from context, . In the later sections, we will introduce classes of particular constraints including contrast, emerging pattern, minimality, and their composite constraints. 


Now, we state our data mining problem considered in this paper. Suppose we fix a constraint $\mathbb C$. 

\subsubsection{Problem:} 
The constrained pattern mining problem w.r.t. constraint $\mathbb C$ 
\begin{itemize}
    \item \textbf{Inputs:} A universe $I$ of items and a labeled database $D = (D_+, D_-)$ over $I$. 
    \item \textbf{Task:} Find all $\mathbb C$-interesting generalized patterns $X \in \mathbb P$ such that $X \in \mathbb C(D \mid \mathbb P)$ on the labeled database $D$. 
\end{itemize}

We remark that the above formulation includes many of previous itemset mining problems by changing the constraint $\mathbb C$. In the remainder of this paper, we consider the pattern mining problem under the constraint of the form $\mathbb{MIN}(\mathbb{CP}[\sigma_+,\sigma_-]\cap \mathbb{GR}[\theta]\cap \mathbb{C}[f, \eta])$, where $f$ is any convex function such as the chi-square constraint $\mathbb{CHI}[\eta]$. 



\subsection{Constraints and scores of pattern}

Let $D = (D_+, D_-)$ be a labeled database. A pattern constraint (or constraint) is a subset $\mathbb C\subseteq \mathbb P$ of patterns. In this paper, we consider the following classes of constraints: 

\subsubsection{Constraint 1. Contrast constraint} 

For any non-negative integers $\sigma_+ \in [0..|D_+|]$ and $\sigma_- \in [0..|D_-|]$, the constraint $\mathbb{CP}[\sigma_+, \sigma_-]$ is defined as follows: any pattern $X$ belongs to $\mathbb{CP}[\sigma_+, \sigma_-]$ if and only if $Sup_+(X) \geq \sigma_+$ and $Sup_-(X) \leq \sigma_-$. 

Members of $\mathbb{CP}[\sigma_+, \sigma_-]$ are called \textbf{contrast patterns}. Members of $\mathbb{CP}[1, 0]$ are called \textbf{jumping emerging patterns}. 

\subsubsection{Constraint 2. Growth rate constraint}

For any non-negative real number $\theta \in [0,\infty]$, the constraint $\mathbb{GR}[\theta]$ is defined as follows: a pattern $X$ belongs to $\mathbb{GR}[\theta]$ if and only if $GR(X\mid D_+, D_-) := Sup_+(X) / Sup_-(X) \geq \theta$. 

Members of $\mathbb{GR}[\theta]$ are called \textbf{emerging patterns patterns}. 

\subsubsection{Constraint 3. Chi-square constraint}







For any non-negative real number $\gamma \in [0,\infty]$, a pattern $X$ belongs to $\mathbb{CHI}[\gamma]$ if and only if $\chi^2(X \mid D_+, D_-) \ge \gamma$, where $\chi^2(X)$ is chi-squared value.

\subsubsection{Constraint 4. Composition of constraint} 

The constraint consisting of all patterns satisfying two constraints $\mathbb C_1$ and $\mathbb C_2$ is represented by their intersection $\mathbb C := \mathbb C_1\cap\mathbb C_2$.

\subsubsection{Constraint 5. Minimality constraint}

Let $\mathbb C\subseteq \mathbb P$ be any constraint. The minimal $\mathbb C$-constraint, denoted by $\mathbb{MIN\:C}$, is the set of all minimal members of   w.r.t.~$\mathbb C$, that is, any pattern $X$ belongs to $\mathbb{MIN\:C}$ if and only if (i) $X$ belongs to $\mathbb C$, and (ii) no proper subset $Y \subset X$ belongs to $\mathbb C$.

\begin{figure}[t]
    \centering
    \includegraphics[width=0.55\linewidth]{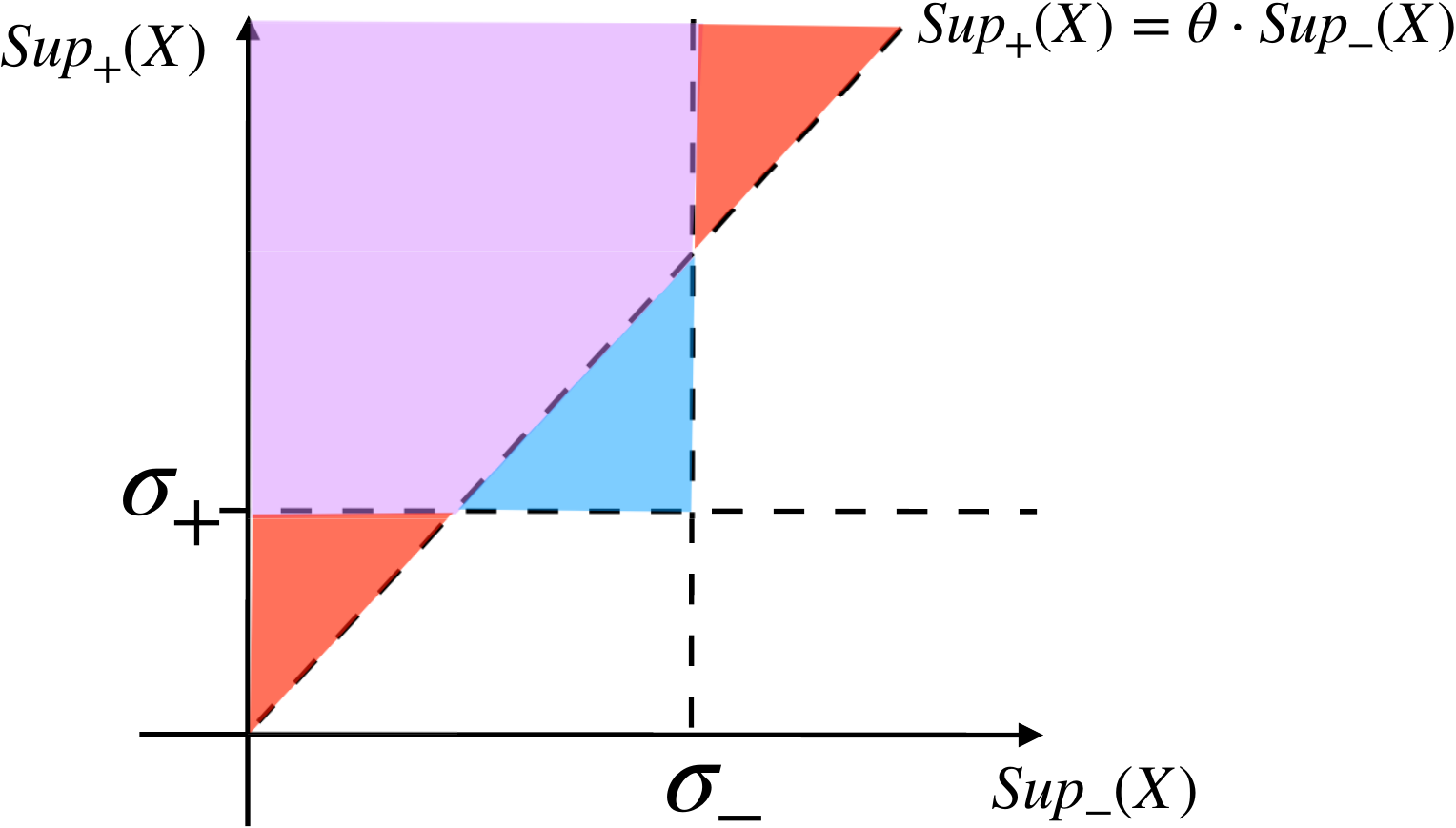}
    \caption{The sub-regions for $\mathbb{CP}[\sigma_+,\sigma_-]$, $\mathbb{GR}[\theta]$, and their composite constraint $\mathbb R :=\mathbb{CP}[\sigma_+,\sigma_-]\cap \mathbb{GR}[\theta]$ as a blue rectangle, a red triangle, and their intersection as a purple pentagon}
    \label{fig:domain}
\end{figure}

\subsubsection{Example of a composite constraint}

For instance, $\mathbb{MIN}(\mathbb{CP}[\sigma_+,\sigma_-]\cap \mathbb{GR}[\theta])$ stands for the class of all minimal patterns that satisfy the contrast constraint w.r.t.~$(\sigma_+,\sigma_-)$ and the growth rate constraint w.r.t.~$\theta$. 

In Fig.~\ref{fig:domain}, we show the sub-regions for $\mathbb{CP}[\sigma_+,\sigma_-]$, $\mathbb{GR}[\theta]$, and their composite constraint $\mathbb R :=\mathbb{CP}[\sigma_+,\sigma_-]\cap \mathbb{GR}[\theta]$ as a blue rectangle, a red triangle, and their intersection as a purple pentagon. 
A minimal pattern is a point within the pentagon $\mathbb R$ that is minimal w.r.t.~set inclusion $\subseteq$.


\section{Algorithms}
\label{sec:alg}

\begin{figure}[t]
    \centering
    \includegraphics[width=0.8\linewidth]{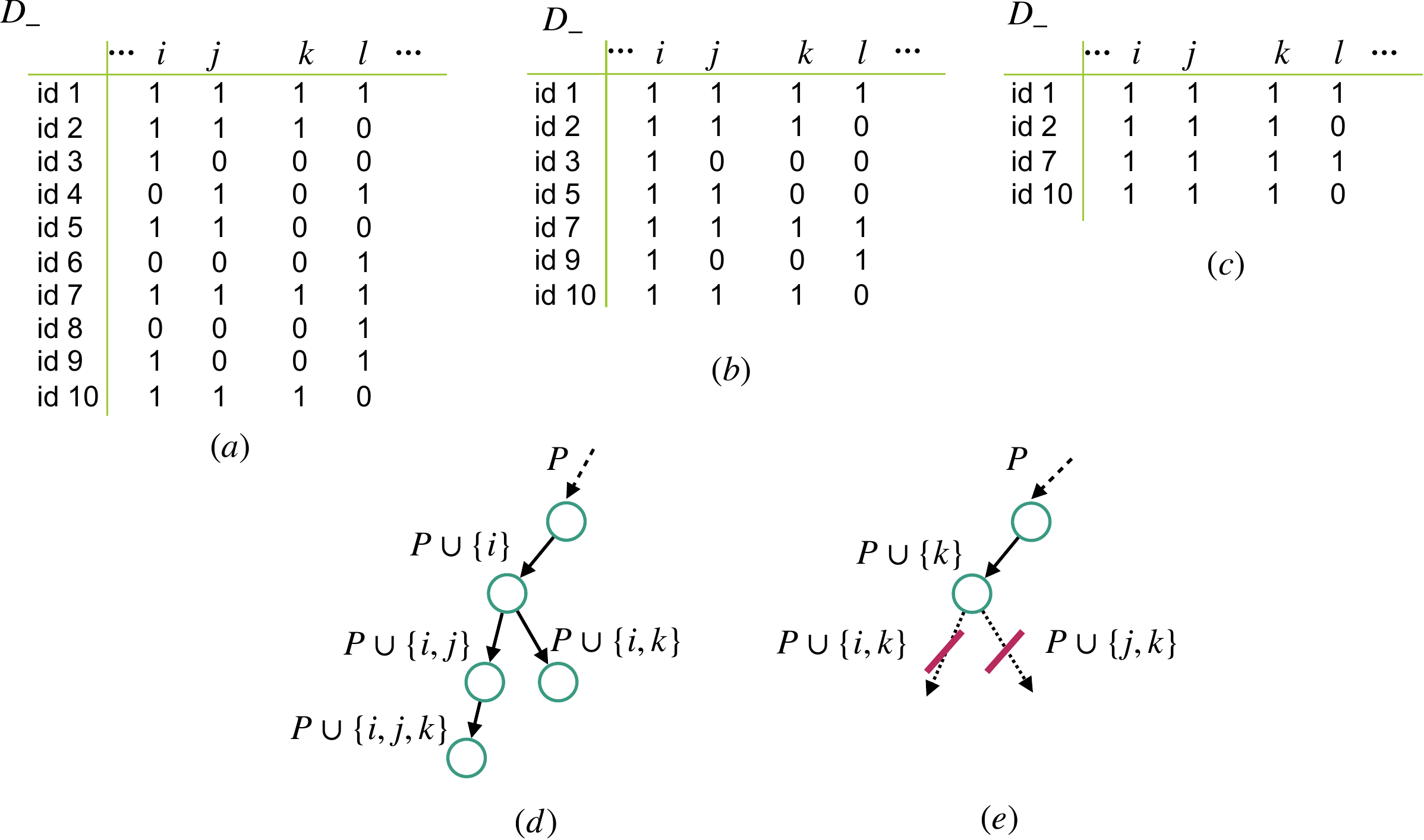}
    \caption{(a): The negative database $D_-$ holding only rows containing a pattern $P$. (b): The negative database $D_-$ holding only rows containing a pattern $P \cup \set{i}$. This can be obtained by deleting the row where the column of $i$ is $0$ in the database in (a). (c): The negative database $D_-$ holding only rows containing a pattern $P \cup \set{k}$. (d): A search tree according to the dictionary order of items. (e): A search tree according to pattern frequency. Since the frequency of $P \cup \set{k}$ that added $k$ to the pattern is less than the frequency of adding other items, $k$ is added preferentially. At this time, $P \cup \set{i,k}$ and $P \cup \set{j,k}$ can be pruned with minimality because because the frequency of negative class does not change from $P \cup \set{k}$.}
    \label{fig:dynamic}
\end{figure}


Our proposed algorithm is based on a depth-first search algorithm such as CP-tree~\cite{Fan_TKDE2006}. During the search, pruning is possible when $ Sup _-(P) = Sup _-(P + \{a\})$ based on the minimality and constraint expression (the correctness of this pruning will be proved later)~\cite{Loekito_KDD2006}. The key idea of our algorithm is to consider the search order for variables that satisfy this pruning rule early and do not perform extra searches.
Figure~\ref{fig:dynamic} shows the basic idea. When variables $i$ and $j$ that are $ Sup_ {D _-} (P + \{i\}) < Sup_ {D _-} (P + \{j\})$ are given as next search candidates, it can be quickly pruned by searching $P + \{i\}$ first. 
Therefore, the search is performed with priority given to $i \in B $ with the smallest $Sup_{D_-}(P + \{i\})$ for the currently searched pattern $P$ and the next variable candidate set $B$.  
This means that instead of searching for variables in a predetermined static order, the search order is determined in a dynamic order according to the currently searched pattern $P$.
In order to implement this, it is important to enabling high-speed counting by always reducing the database to contain records only with the pattern $P$ currently being searched.  There is an enumeration of contrast patterns using ZDD in this method using pruning and database reduction. However, because ZDD is constructed with only the static variable order determined in advance due to its data structure, it is difficult to extend to dynamic order.
In this section, we first show the pruning rules used in this algorithm, and then give a data structure that uses dancing links to efficiently execute dynamic variable ordering.

\subsection{Pruning rules}

In this subsection, we explain the types of pruning used in the proposed algorithm. 

\subsubsection{Pruning 1: for minimality constraint}

This is one of the obvious pruning. If we find a pattern $P$ that satisfies constraints other than minimality constraints, we do not need to search for patterns $Q \supset P$.

\subsubsection{Pruning 2: for $EP$ and $CP$ by lower bound}

Let $P \in 2^I$ any pattern at the current iteration. We define the set of descendants by $Desc(P) := \set{ Q \in 2^I \mid P \subseteq Q }$.  
We define the upper bound and the lower bound of the values of $f$ on all descendants of $P$ by
\begin{align*}
  GUB_{D_+, D_-}[f](P)
  &:= GUB[f](Desc(P))
  \\&= \max \set{ f(Q) \mid Q \in 2^I, P \subseteq Q }.
\end{align*}
\begin{align*}
  GLB_{D_+, D_-}[f](P)
  &:= GLB[f](Desc(P))
  \\&= \min \set{ f(Q) \mid Q \in 2^I, P \subseteq Q }.
\end{align*}

If $B$ is a set of all unsearched items, the lower bound of negative frequency is $lb_occ_- =  Sup _- (P \ cup B)$ due to the monotonicity of frequency. 
Similarly, the upper bound of the positive frequency is $ ub_occ _ + = Sup _ + (P) $.
Let $f = GR = Sup _ + (P) / Sup _- (P) $, that is, under the condition of the growth rate constraint, the upper bound on all descendants is $ GUB_ {D_ +, D _-} [GR] (P) \leq  ub_occ _ + / lb_occ _- $. Thus, pruning is possible when this upper bound is smaller than the given parameter $\theta$. The pruning using the lower bound of the negative  frequency and the upper bound of the positive frequency can be similarly used for the contrast constraint and the chi-square value constraint.



\subsubsection{Pruning 3: Safe pruning for minimal constrained patterns based on negative conservative elements}  

We show the soundness of a pruning strategy using the negative occurrences. 



\begin{definition}[Condition C1']
For any pattern $P$, and any $a \in I\cup \neg I$, if $Sup_-(P) = Sup_-(P+a)$ then the implication $P+a \in \mathbb C \implies P \in \mathbb C$ holds. 
\end{definition} 

\begin{definition}[Condition C2'] 
If $Occ_+(P) = Occ_+(Q)$ and $Occ_-(P) = Occ_-(Q)$, then the equivalence $P \in \mathbb C \iff Q \in \mathbb C$ holds.
\end{definition} 


Now, we have the following lemma and theorem.

\begin{lemma}[Soundness of pruning for $\mathbb C$-patterns]\label{lemma:pr3}: Suppose that a constraint $\mathbb C$ satisfies Conditions C1' and C2' above. Let $P$ be any pattern and any item $a \notin P$. Suppose that $Sup_-(P) = Sup_-(P+a)$. For any pattern $Z$ that is an extension of $P+a$, where $P+a \subseteq Z$, $Z \in \mathbb C \implies Z\setminus\{a\} \in \mathbb C$. 
\end{lemma}








\begin{theorem}[Soundness of pruning for minimal $\mathbb C$-patterns]: Suppose that a constraint $\mathbb C$ satisfies Conditions C1' and C2' above. Suppose that $Sup_-(P) = Sup_-(P+a)$. Then, any extension $Z$ of pattern $(P+a)$ never  satisfy the minimality constraint $\mathbb{MINC}$ w.r.t. $\mathbb{C}$. That is, for any $Z$, the condition $P+a \subseteq Z$ implies that $Z \notin \mathbb{MINC}$. 
\end{theorem}

\begin{proof}
We assume the conditions of C1' and C2', and  that a pattern $P$ satisfies $Sup_-(P) = Sup_-(P+a)$. Now, suppose to contradict that $Z \in \mathbb{MINC}$ for some (possibly identical) extension $Z$ of $(P+a)$. Since $\mathbb{MINC}\subseteq \mathbb{C}$, we  have $Z \in \mathbb{C}$. Then, it immediately follows from Lemma~\ref{lemma:pr3} that $Z\setminus\{a\} \in \mathbb{C}$. Since $a \in Z$, the set $Z\setminus\{a\}$ is a strict subset of $Z$. Hence, $Z$ cannot be minimal in $\mathbb{C}$, i.e., $Z \notin \mathbb{MINC}$. 
\end{proof}

At any unsuccessful iteration on $P$ such that $P \notin \mathbb C$, if the condition $Sup_-(P) = Sup_-(P+a)$ holds, then we can prune all the descendants of  $(P+a)$, and then backtrack to the parent $P$.

\subsection{Dynamically reducible binary matrices}

In this subsection, we propose an  representation of a binary matrix, called, dynamically reducible binary matrix (DRMX), 
which allows efficient modification and undo operations on a transaction database necessary to dynamic item ordering during backtrack search for candidate patterns. To achieve this goal, we employ the dancing link data structure of Knuth~\cite{Knuth_2000}. 

\begin{definition}[D1]
The DRMX data structure $\mathbb M$ stores a transaction database $M = (T, I, R)$ supports the following operations, where we refer to a tuple and an item as a row and a column, respectively. 
\end{definition}


\begin{itemize}
    \item  $\mathbb M := DRMX.create(M)$: Create a new DRMX storing a given transaction database $D$. 
    \item  $\mathbb M.deleteRow(i)$: Remove the row with rid $i$ from the matrix. 
    \item  $\mathbb M.deleteColumn(j)$: Remove the column with cid $j$ from the matrix. 
    \item $\mathbb M.checkpoint()$: Push the current state of the matrix on the undo-stack
    \item   $\mathbb M.undo(i)$: Pop $i$-times from the undo-stack, and then recover the status of the matrix $M$ at the time of the checkpoint
    \item  $\mathbb M.countRows(P)$: Return the number of rows where $j=1$ for all column $j\in P$ in the matrix $M$. If $P = \emptyset$ then return the number of rows in the matrix $M$.
\end{itemize}

Dancing links can perform these operations efficiently. In the next subsection we describe our algorithm in pseudo  code using these operations.

\subsection{Pseudo codes of our algorithm}
Pseudo codes of our algorithm is shown in Algorithm~\ref{alg:main} and Algorithm~\ref{alg:mine}. In the algorithm~\ref{alg:main}, the solution candidates are mined on the line~3, and then the minimal solution is extracted on the fourth line. A method using BDD has been proposed for narrowing down the minimum solution~\cite{Toda_EA2013}. Algorithm~\ref{alg:mine} is the main mining algorithm.

\IncMargin{1em}

 \SetKwInOut{Input}{input}
 \SetKwInOut{Output}{output}
 \SetFuncSty{textrm}
 \SetCommentSty{textrm}
 \SetKwFunction{MiningMCP}{{\scshape MiningMCP}}
 \SetKwFunction{FindCandidates}{{\scshape FindCandidates}}
 \SetKwFunction{ExtractMinimalPatterns}{{\scshape ExtractMinimalPatterns}}
 \SetKwProg{myfunc}{}{}{}
\begin{algorithm}[th]
 \label{alg:main}
\caption{Main function for mining $\mathbb{MINCP}$}
\Input{A pair $D_+, D_- \subseteq 2^{I}$ of positive and negative datasets represented in the DRMX data structure, a tuple $\Theta = (\sigma_+, \sigma_-, \theta, \gamma)$ of  mining parameters (See above for the meaning of symbols). }
\Output{The set $MCP \subseteq 2^{I}$ of all and only minimal constrained patterns.}
\myfunc{\MiningMCP{$D, I, \Theta$}}{
    $(D_+, D_-) \gets DRMX.create(D)$ \;
    $CP \gets \FindCandidates(\emptyset, I, \mathbb  D_+, D_-, \Theta)$ \;
    $MCP \gets \ExtractMinimalPatterns(CP)$ \;
    \Return $MCP$ \;
}
\end{algorithm}
\DecMargin{1em}


\newcommand{\deli}[1]{\setminus\set{#1}}

\IncMargin{1em}
\begin{algorithm}[h!]
 \small
 \label{alg:mine}
 \SetKwInOut{Input}{input}
 \SetKwInOut{Output}{output}
 \SetFuncSty{textrm}
 \SetCommentSty{textrm}
 \SetKwFunction{FindCandidates}{{\scshape FindCandidates}}
 \SetKwProg{myfunc}{}{}{}
 \caption{An algorithm for finding candidates for   minimal contrast patterns in $\mathbb{MINCP}/$ under the following mining parameters: $\sigma_+, \sigma_-$: positive and negative minimum positive support thresholds, $\theta$: minimum growth-rate threshold, $\gamma$: minimum $\chi^2$-value threshold. }

\myfunc{\FindCandidates{$P, B, D_+, D_-, \Theta$}}{
  $CP \gets \emptyset$ \; 
    $(occ_+, occ_-) \gets (D_+.countRow(), D_-.countRow())$  \tcp*{positive and negative frequencies}
  \tcp{Pruning1: Pruning with minimal constraints}
    \If{$isCP(occ_+, occ_-, \Theta) = true$}{
        $CP \gets CP \cup \set{P}$ \tcp*{discover patterns}
        \Return{CP} \; 
    }

    \tcp{Pruning2: Pruning using the lower bound of the negative frequency of all descendants of $P$}
    $lb\_occ_- \gets D_+.countRow(\emptyset)$ \tcp*{the lower bound for negative frequency of $P$}
      \If{ $(lb\_occ_- > \sigma_-)$ }{
        \tcp{Prunes all descendants of $P \ cup \ set {i}$}
        return \;
      }

    \tcp{Pruning 3: Pruning EP using the upper bound of the GR of all descendants of $P$}
    $ub\_occ_+ \gets D_-.countRow(B)$ \tcp*{the upper bound of positive frequency of $P$}
    $ub\_gr \gets ub\_occ_+/lb\_occ_-$ \tcp*{Maximum GR of all descendants of $P$}
    \If{  $ub\_gr < \theta)$ }{ 
      return \tcp*{Prunes all descendants of $P$}
    }

    \tcp{Pruning 4: Database reduction based on minimal constraints}
    \For{each $i \in B$}{
      \If{$(D_-.countRow() = D_-.countRow(i))$}{
        $B \gets B \setminus i$ \;
        $D_k.deleteColumn(i)$, $\forall k\in\set{+,-}$ \;
      }
    } 
    
    \If{$B = \emptyset$}{
        \Return{CP}
    }

    \tcp{Dynamic Ordering based on the frequency of $P\cup\set{i}$ on $D_-$}
    $i_* \gets \argmin_{i \in B} ( D_-.countRow(i ))$
    \tcp*{Select the least frequent item on $D_-$}

    Record the current snapshot of $D_k$ as a checkpoint $\tau$\; 
    \tcp{Branch 0}
    $D_k \gets D_k\deli{i_*}$, $\forall k\in\set{+,-}$ \tcp*{fast implementation by DRMX}
    $CP \gets CP \; \cup$ \FindCandidates{$P, B\deli{i_*}, D_+, D_-, \Theta$} \;
    \vspace{1em}
    \tcp{Branch 1}
    $D_k \gets \set{ t \in D_k \mid P\cup\set{i_*} \subseteq t }$, $\forall k\in\set{+,-}$ \; 
    $CP \gets CP \;\cup$ \FindCandidates{$P\cup\set{i_*},  B\deli{i_*}, D_+, D_-, \Theta$}  \; 
    Undo the modification of $D_k$ at $\tau$\; 
    \Return{CP}
}
\end{algorithm}
\DecMargin{1em}
\section{Experimental results}
\label{sec:exp}


We examine the following three experimental results in this section.
The first is a speed comparison between the heuristic with various static variable orders and the proposed dynamic variable order.
Second, speed comparison with other methods.
Finally, we evaluate the classification model using the patterns that we actually mined.
We implemented our algorithm using C++.
All CPU time is measured on a Linux workstation with Intel Xeon E5-2680 v2 2.80GHz CPU with 400GB memory.

\subsection{Experiments 1: Effectiveness of dynamic ordering}

In this experiment, we investigate the speed difference between the static variable order and the dynamic variable order in our proposed method.
Table~\ref{tbl:dataset-enum} gives the datasets used for performance evaluation in this and the next subsections.
They are all from the CP4IM dataset\footnote{https://dtai.cs.kuleuven.be/CP4IM/datasets/} with more than 50 items and more than 200 examples.
The column ``density'' shows the average percentage of the number of items in an example over the number of all items.
Mushroom and Splice-1 are relatively sparse, German-credit is moderate, and the rest are dense.
The columns ``\#JEPs'' and ``\#SJEPs'' indicate the number of jumping emerging patterns and strong jumping emerging patterns, respectively, when minimum support threshold is set to 0.02 times the number of positive examples.

\begin{table}[!ht]\centering
\caption{Datasets used in Experiments 1 and 2.}
\label{tbl:dataset-enum}
\begin{tabular}{crrrrr}
name & \#item & \#example & density & \#JEPs & \#SJEPs \\
\hline
Mushroom & 119 & 8124 & 18\% & 21574290 & 1353 \\
Splice-1 & 287 & 3190 & 21\% & 377330 & 179810 \\
German-credit & 112 & 1000 & 34\% & 2410029163 & 148303 \\
Kr-vs-kp & 73 & 3196 & 49\% & 129786095160 & 7283 \\
Hypothyroid & 88 & 3247 & 50\% & 40807701172704 & 1966 \\
Anneal & 93 & 812 & 45\% & 34803198050304 & 3906 \\
Heart-cleveland & 95 & 296 & 47\% & 29701186840434 & 946235 \\
Australian-credit & 125 & 653 & 41\% & 261786633471699 & 2057646 \\
Audiology & 148 & 216 & 45\% & \textit{unknown} & 2858 \\
\end{tabular}
\end{table}

\begin{table}[!ht]
\centering
\caption{Datasets used in Experiment 3.}

\begin{tabularx}{\columnwidth}{cXXX}
name & \#sample & {\#feature (not binarized)} & \#target class \\ \hline
Banknote Authentication & 1372 & 5 & 1 \\
Breast Tissue& 106 & 10 & 6 \\
Glass Identification & 214 & 10 & 6\\
Iris & 150 & 4 & 3 \\
Wireless Indoor Localization (Wifi) & 2000 & 7 & 4 \\
Yeast & 1484 & 8 & 9 \\
\end{tabularx}
\label{tbl:dataset-pred}
\end{table}

\begin{figure}[t]\centering
\includegraphics[width=.32\columnwidth]{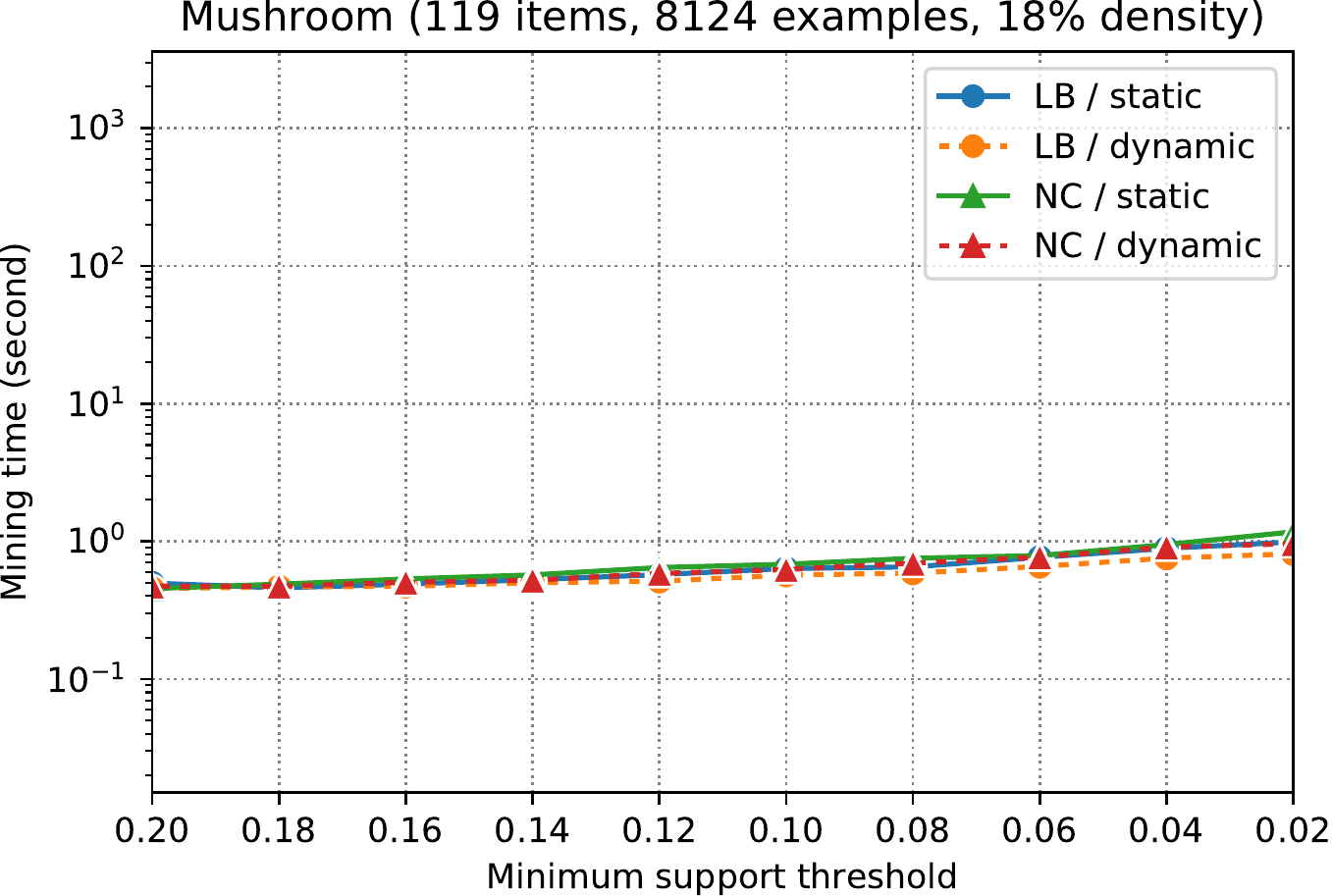}\hfill
\includegraphics[width=.32\columnwidth]{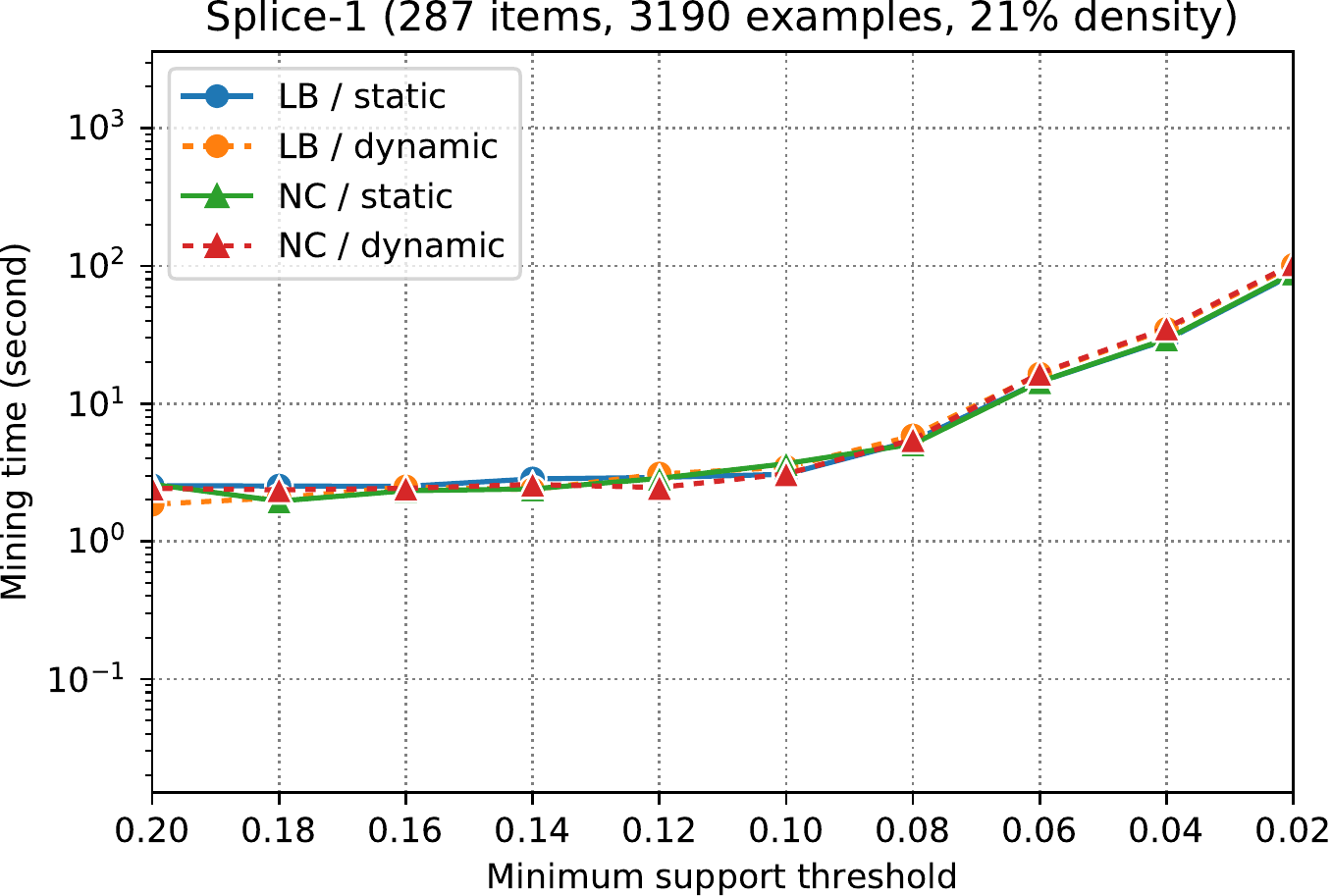}\hfill
\includegraphics[width=.32\columnwidth]{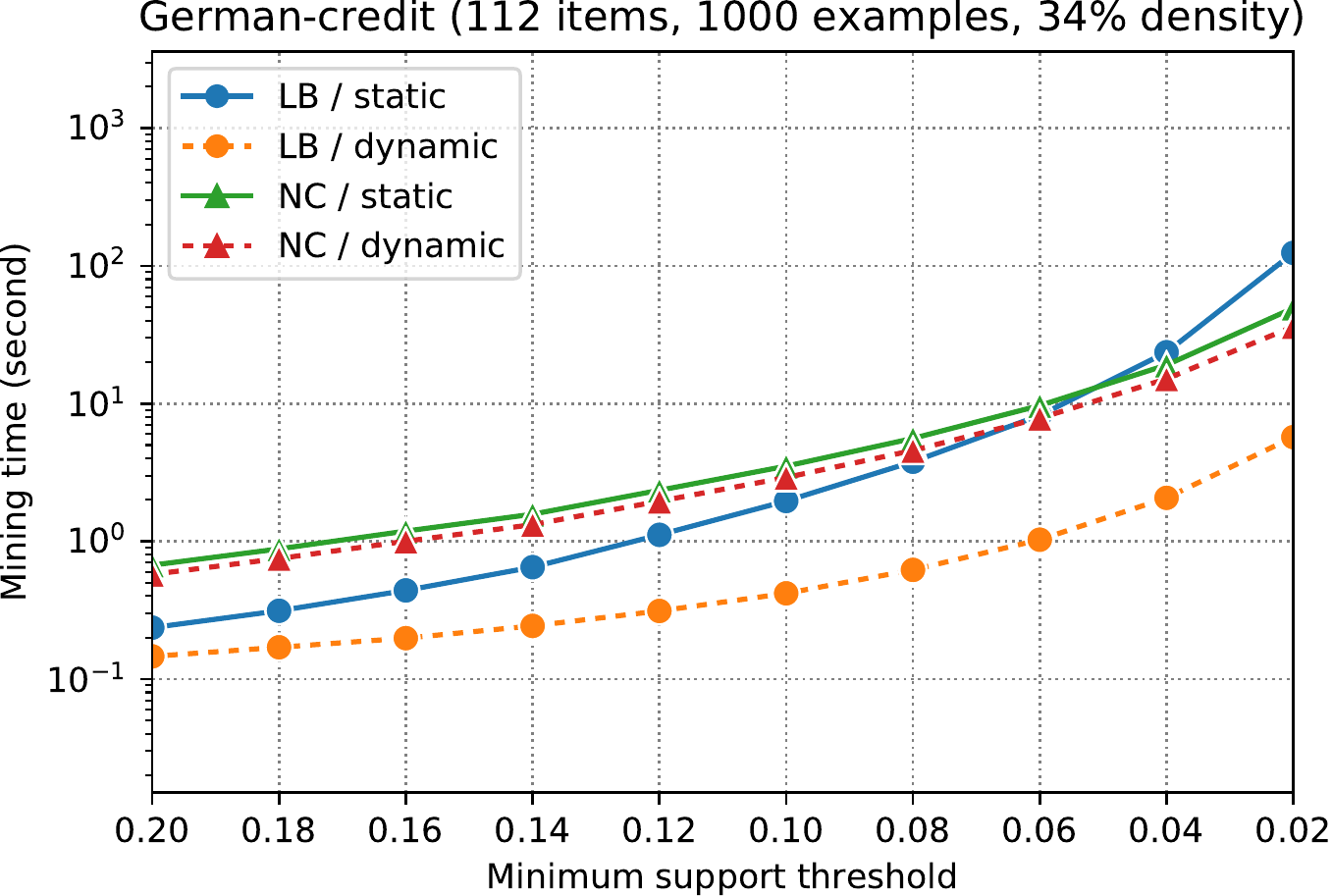}\\\medskip
\includegraphics[width=.32\columnwidth]{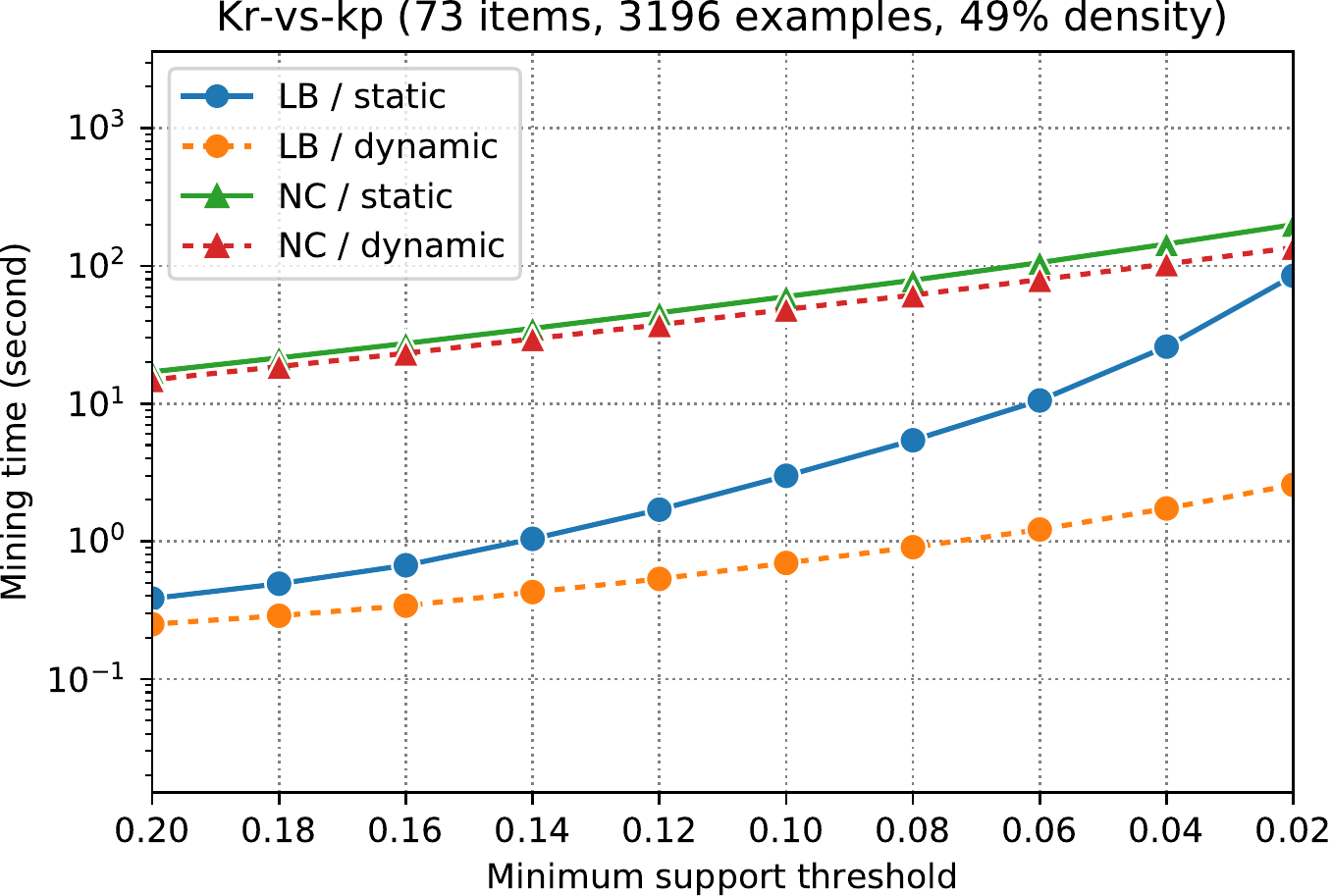}\hfill
\includegraphics[width=.32\columnwidth]{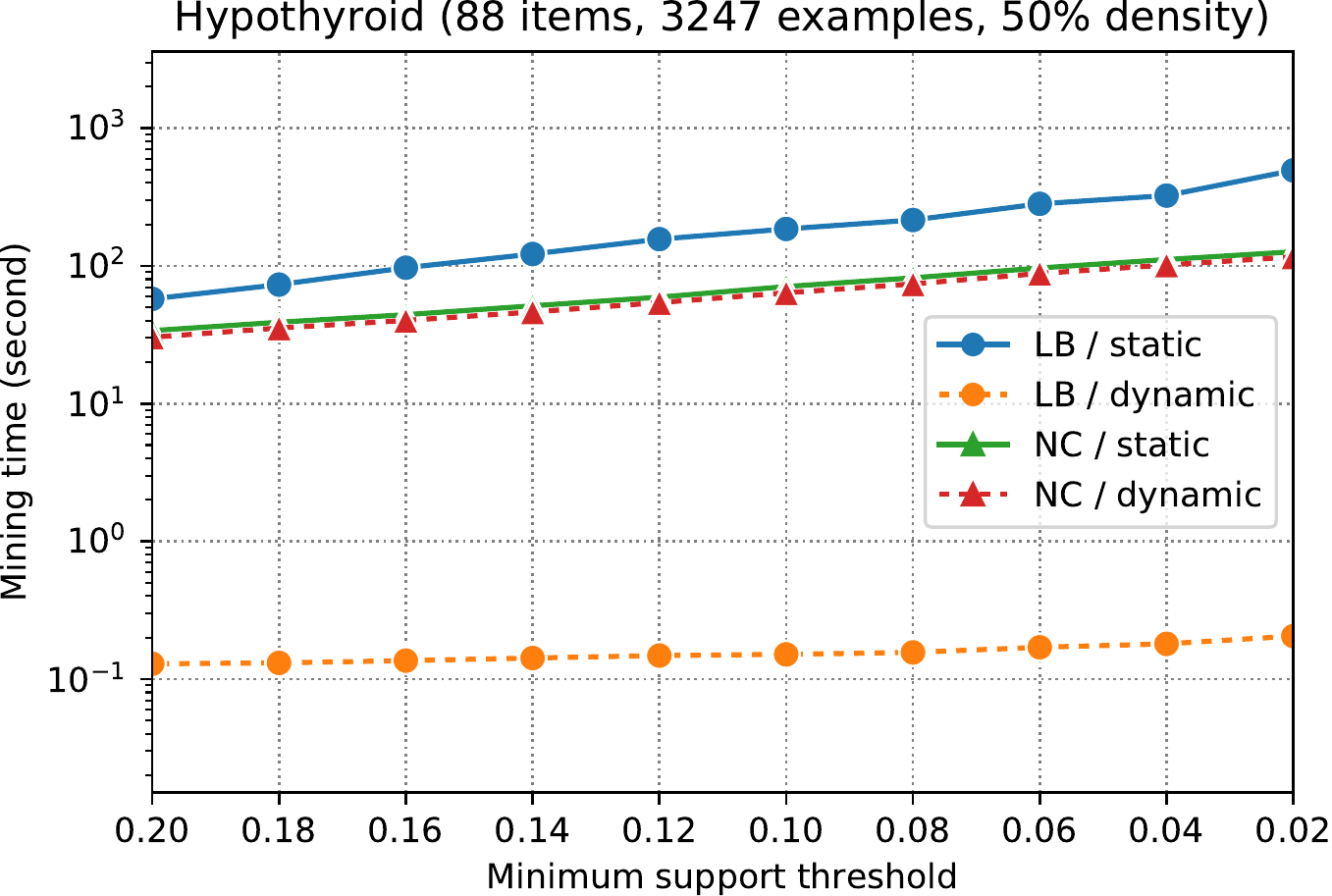}\hfill
\includegraphics[width=.32\columnwidth]{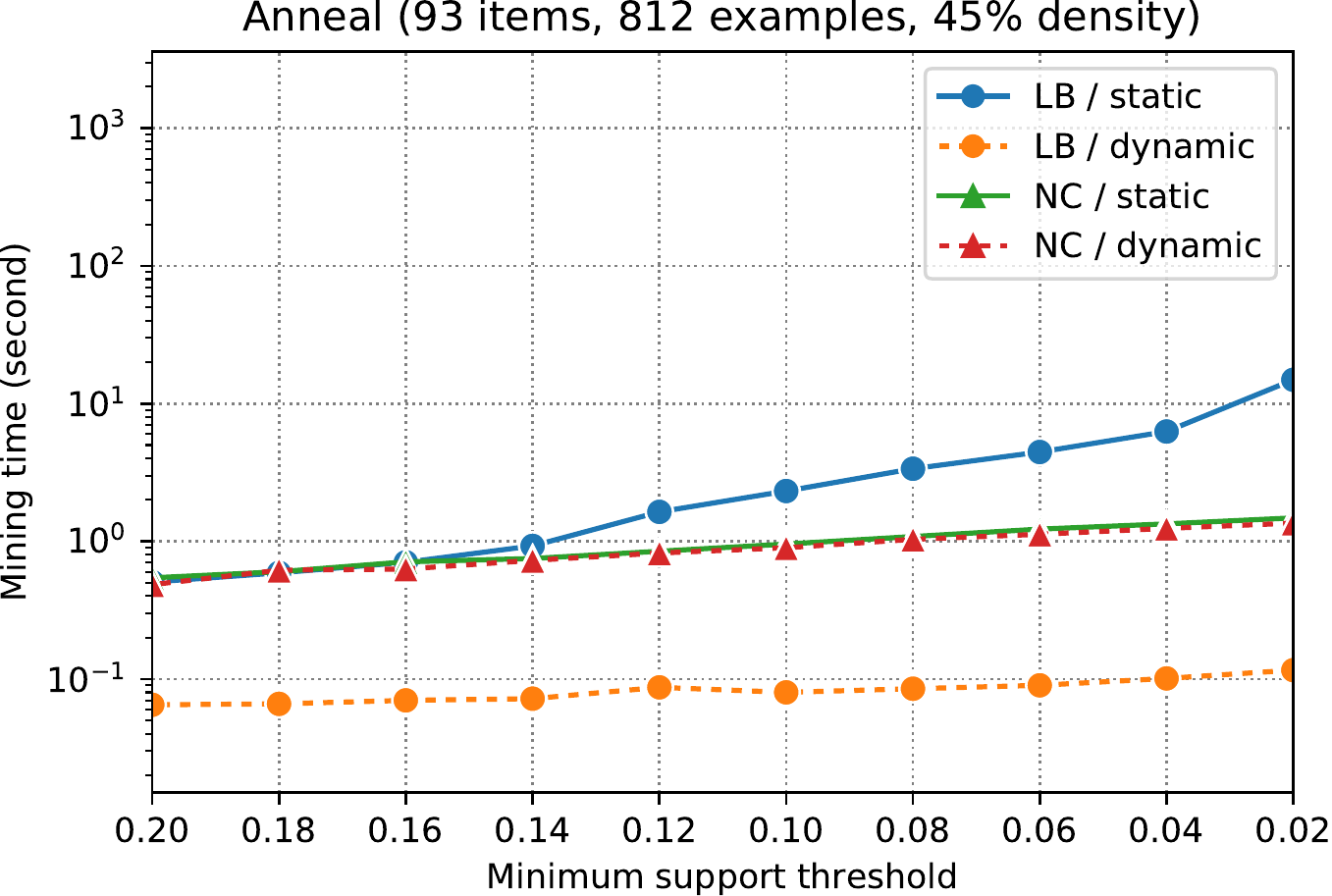}\\\medskip
\includegraphics[width=.32\columnwidth]{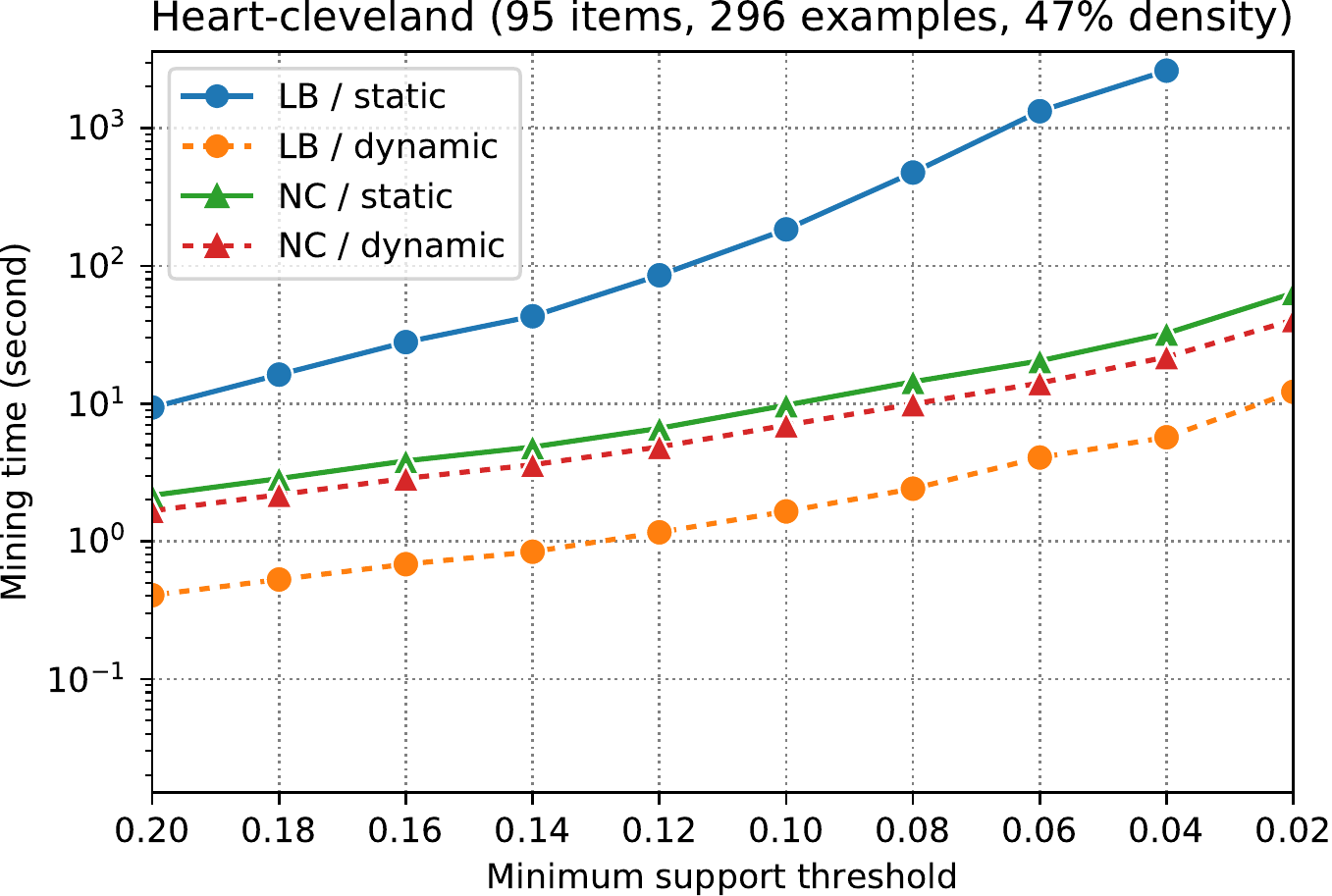}\hfill
\includegraphics[width=.32\columnwidth]{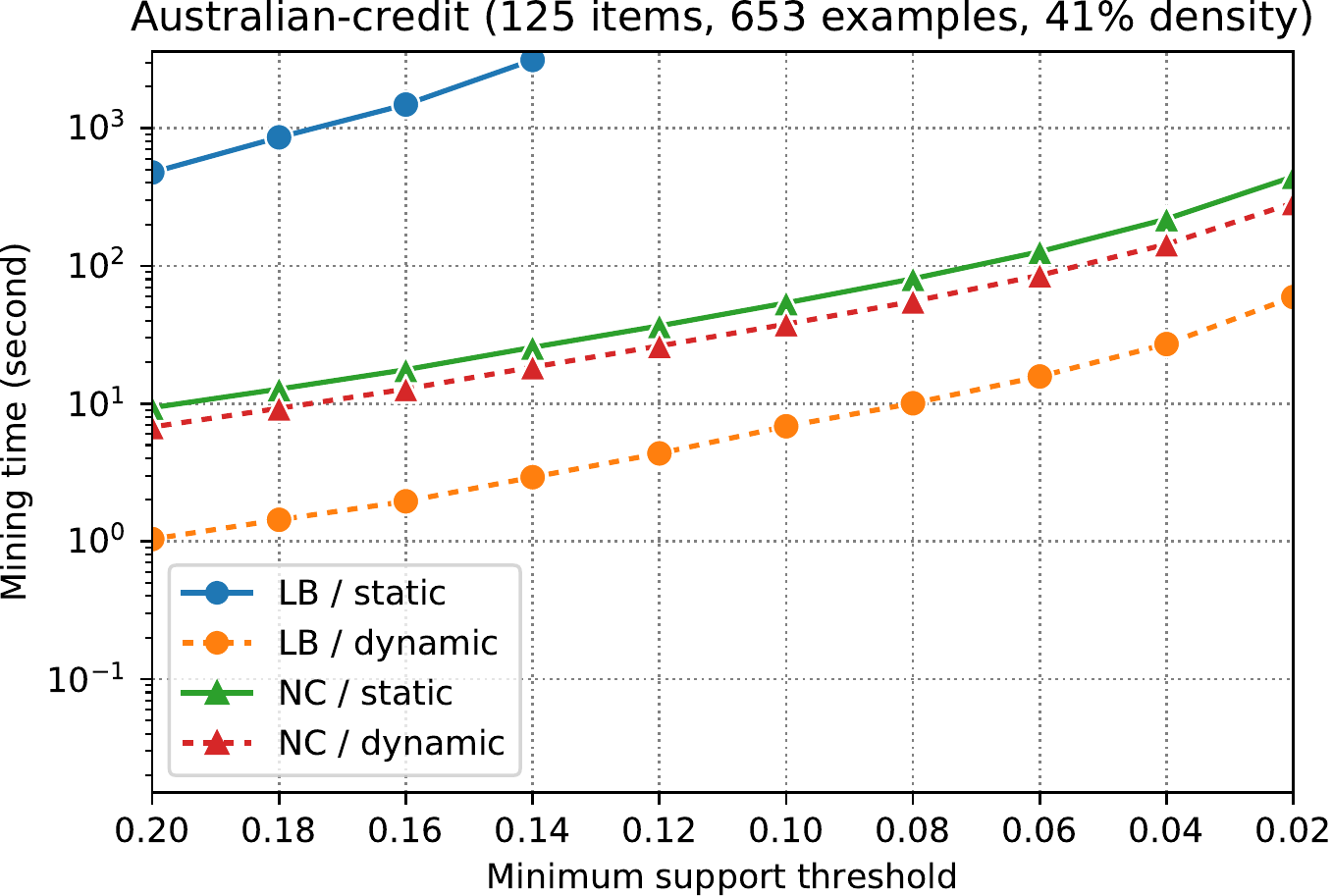}\hfill
\includegraphics[width=.32\columnwidth]{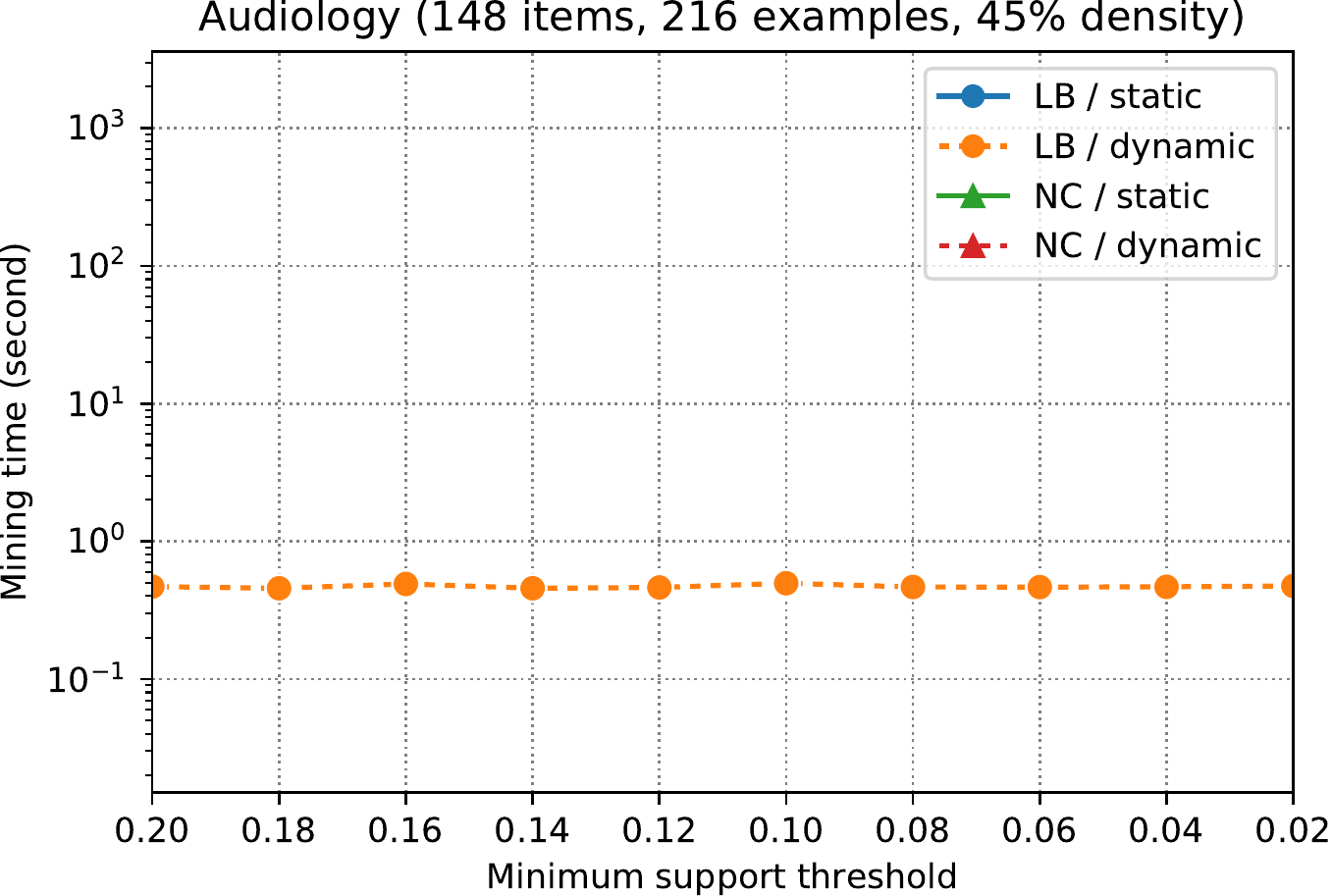}\\
\caption{Comparison of the mining time for minimal emerging patterns with $\theta=9$, using combinations of the pruning rule 2 (LB) or 3 (NC) and static or dynamic ordering.}
\label{fig:pruning-and-ordering}
\end{figure}

\begin{figure}[h!]\centering
\includegraphics[width=.32\columnwidth]{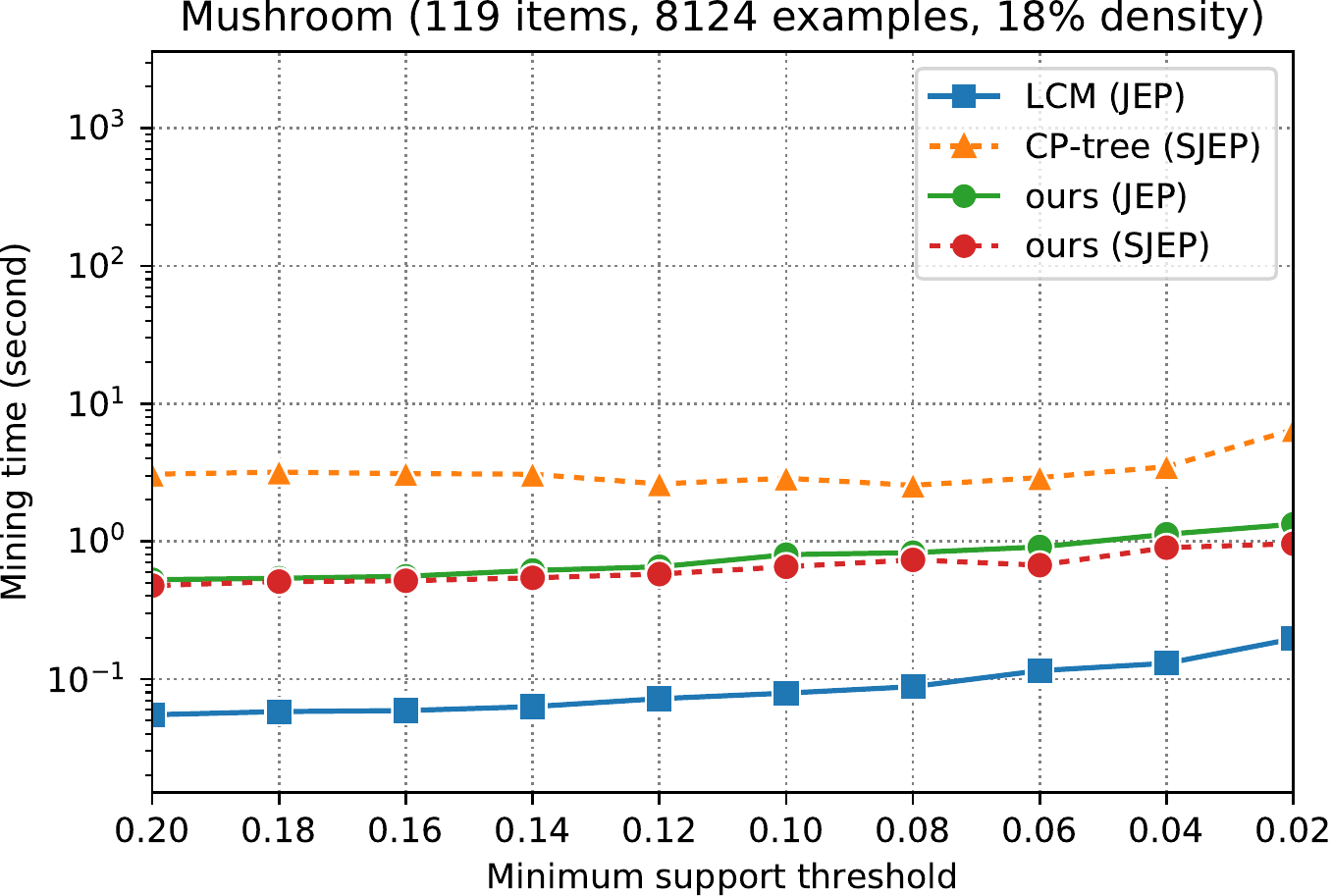}\hfill
\includegraphics[width=.32\columnwidth]{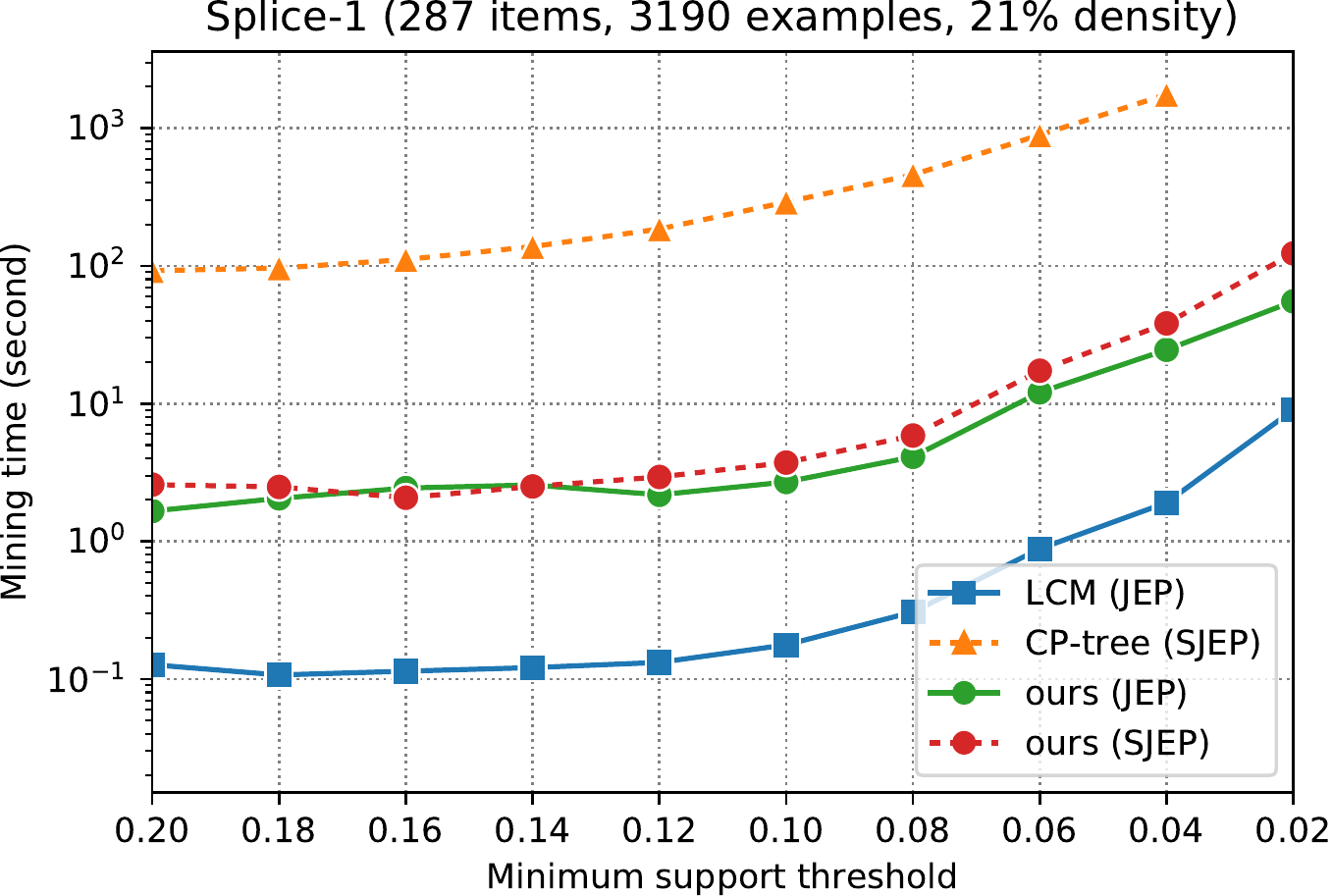}\hfill
\includegraphics[width=.32\columnwidth]{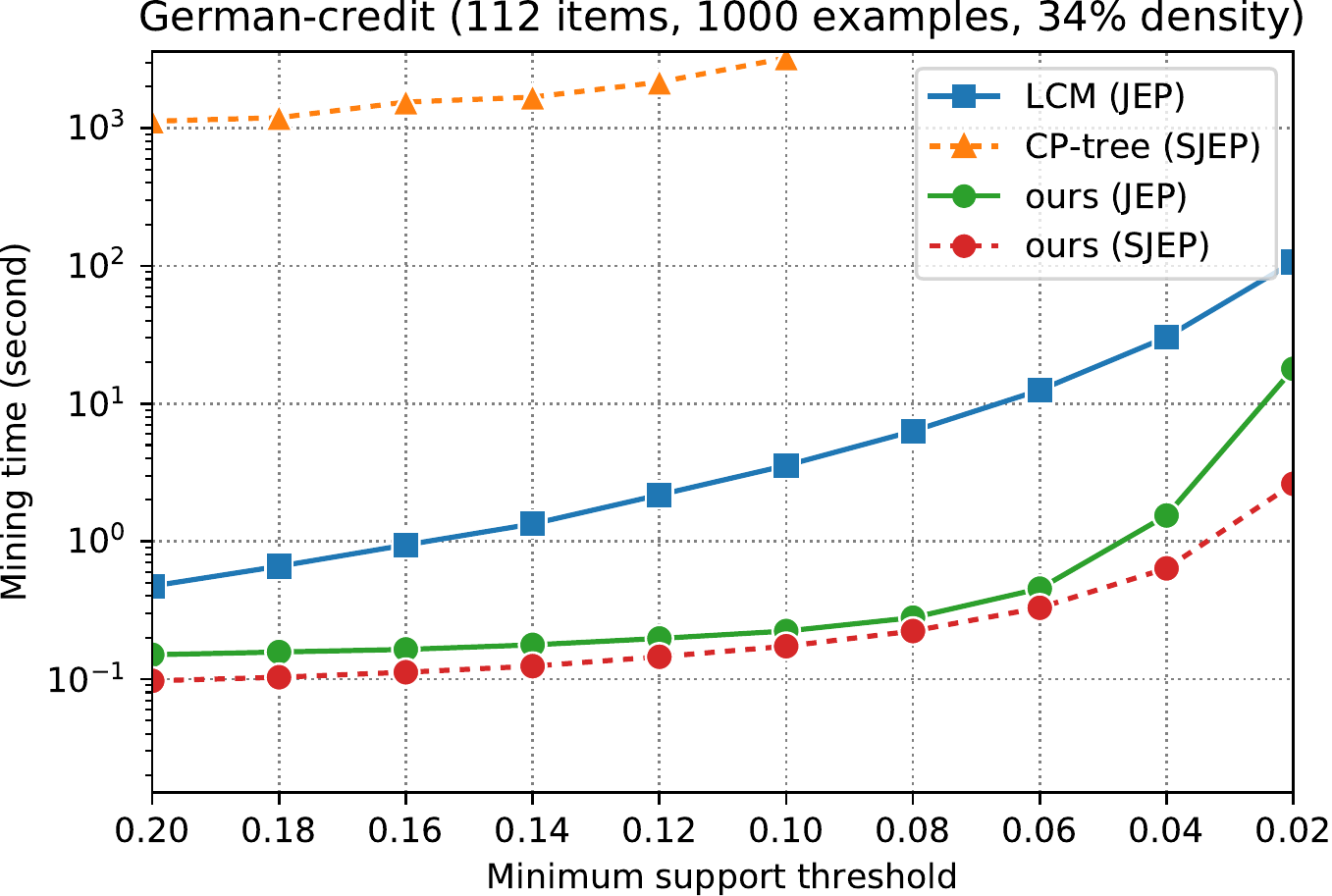}\\\medskip
\includegraphics[width=.32\columnwidth]{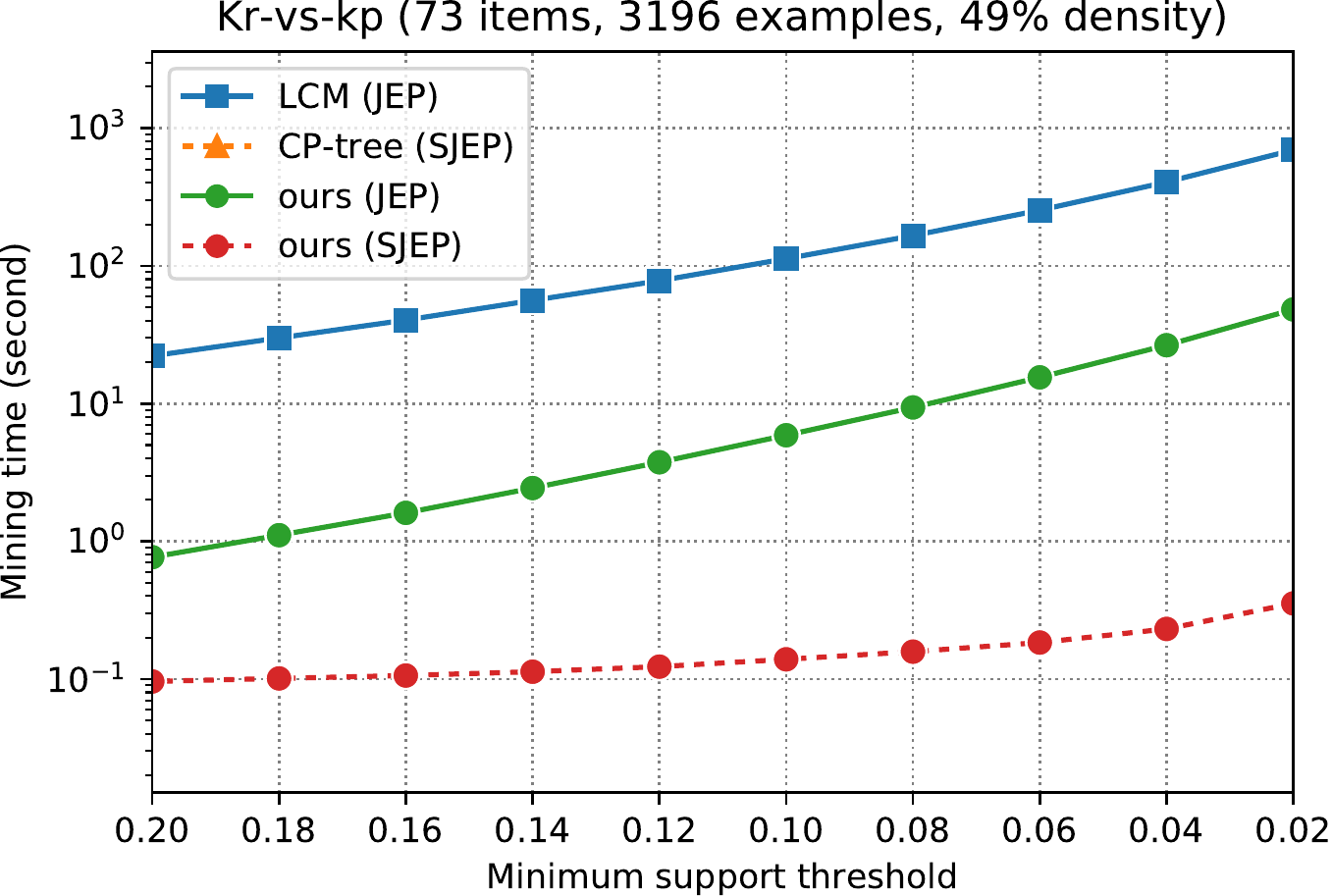}\hfill
\includegraphics[width=.32\columnwidth]{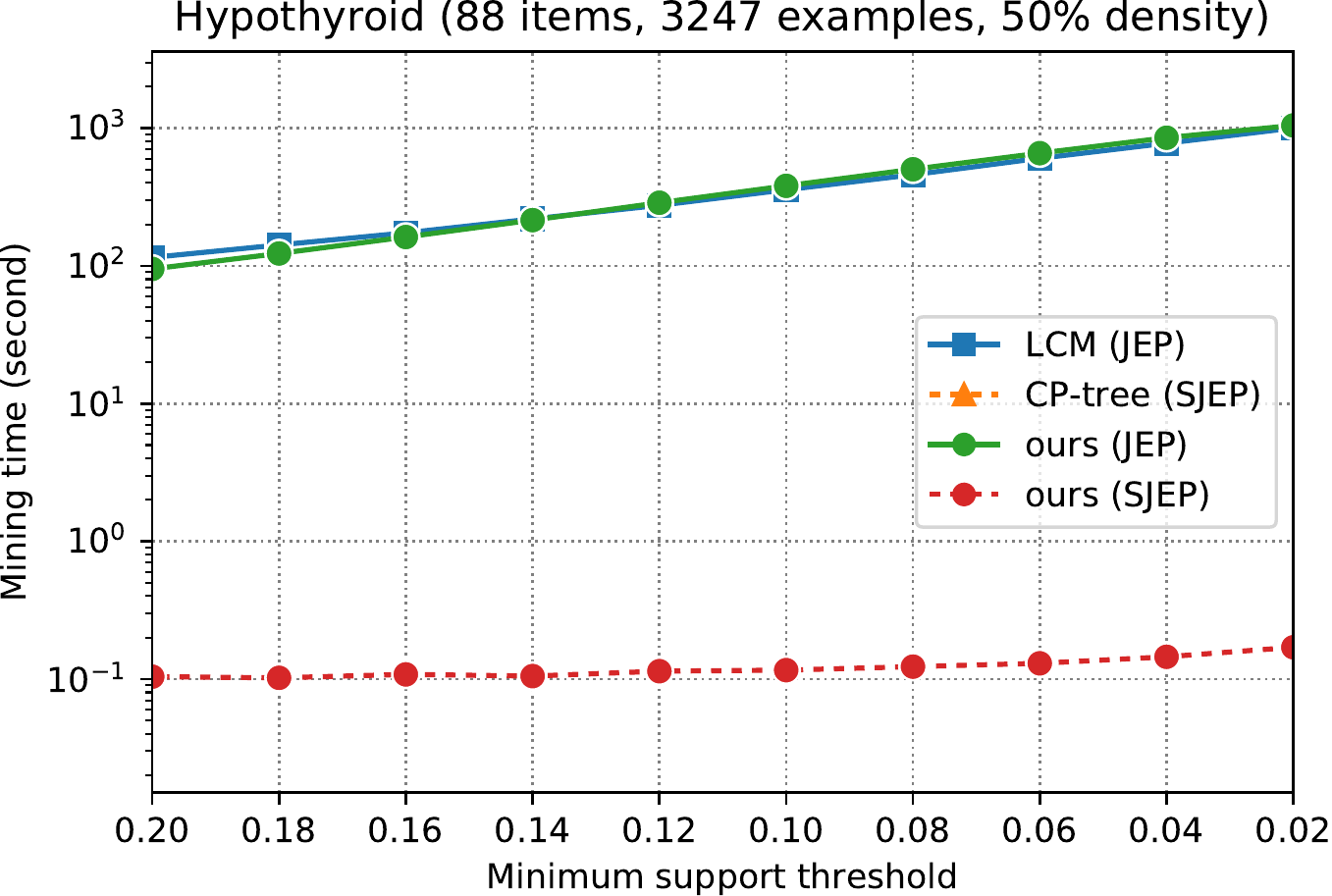}\hfill
\includegraphics[width=.32\columnwidth]{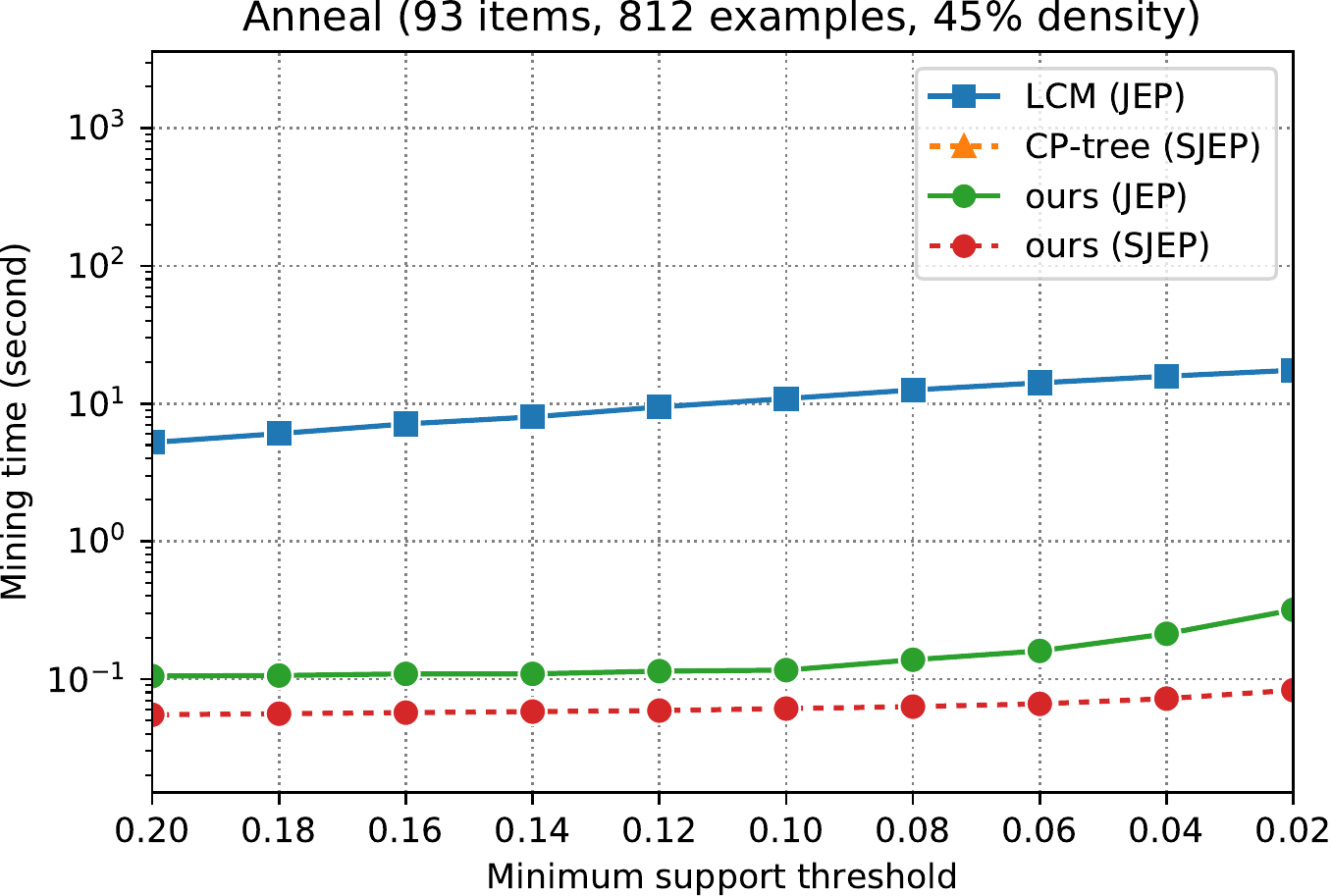}\\\medskip
\includegraphics[width=.32\columnwidth]{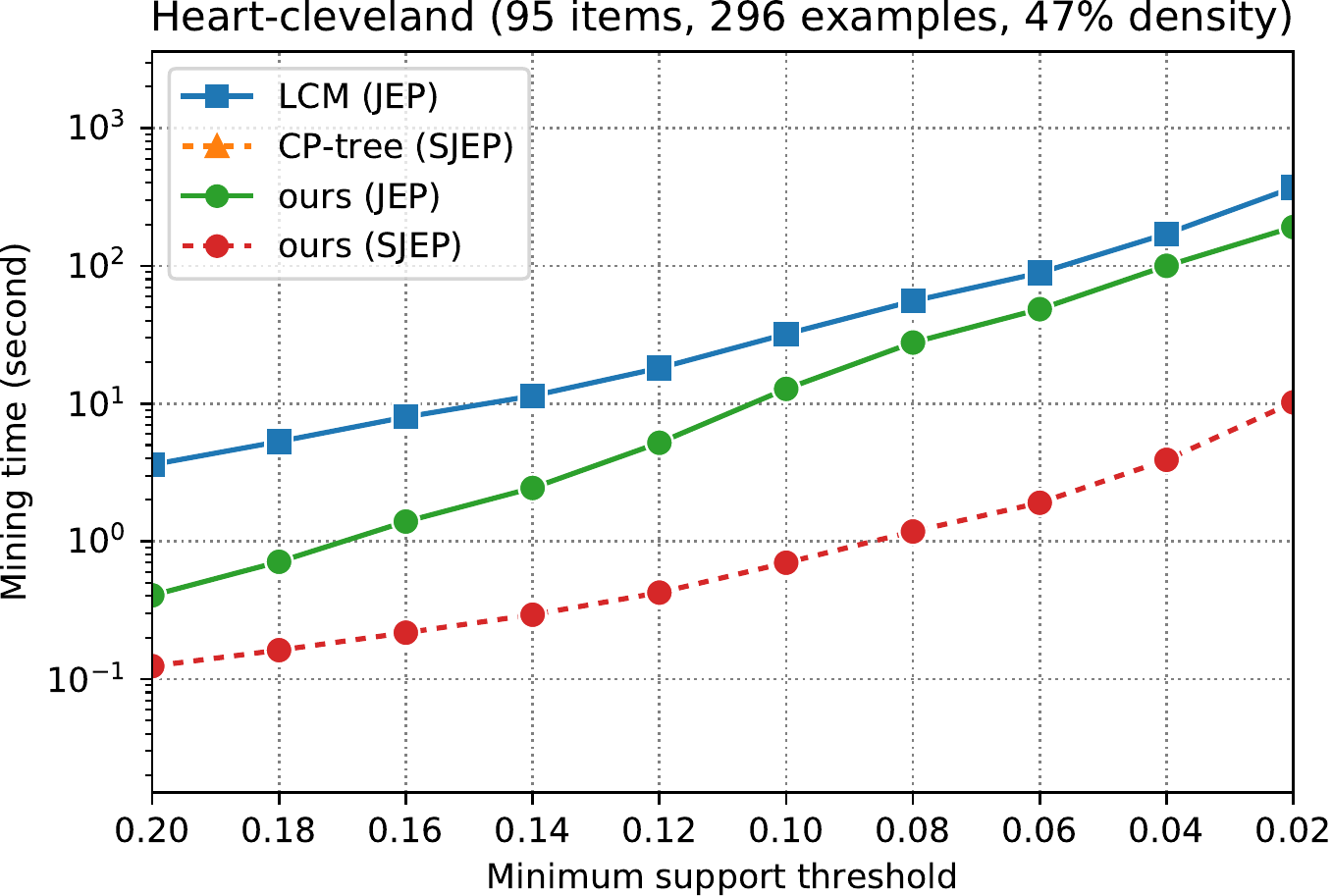}\hfill
\includegraphics[width=.32\columnwidth]{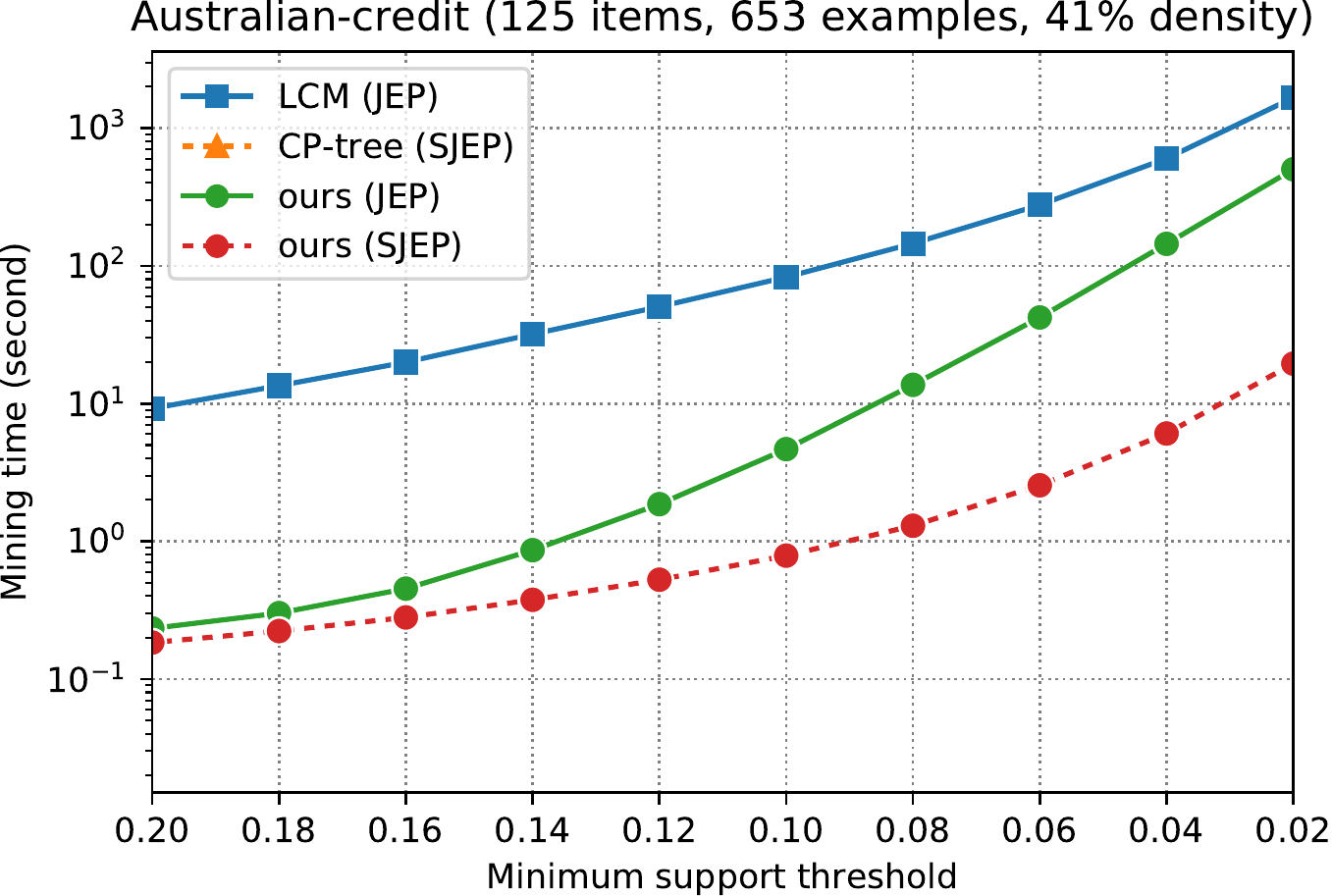}\hfill
\includegraphics[width=.32\columnwidth]{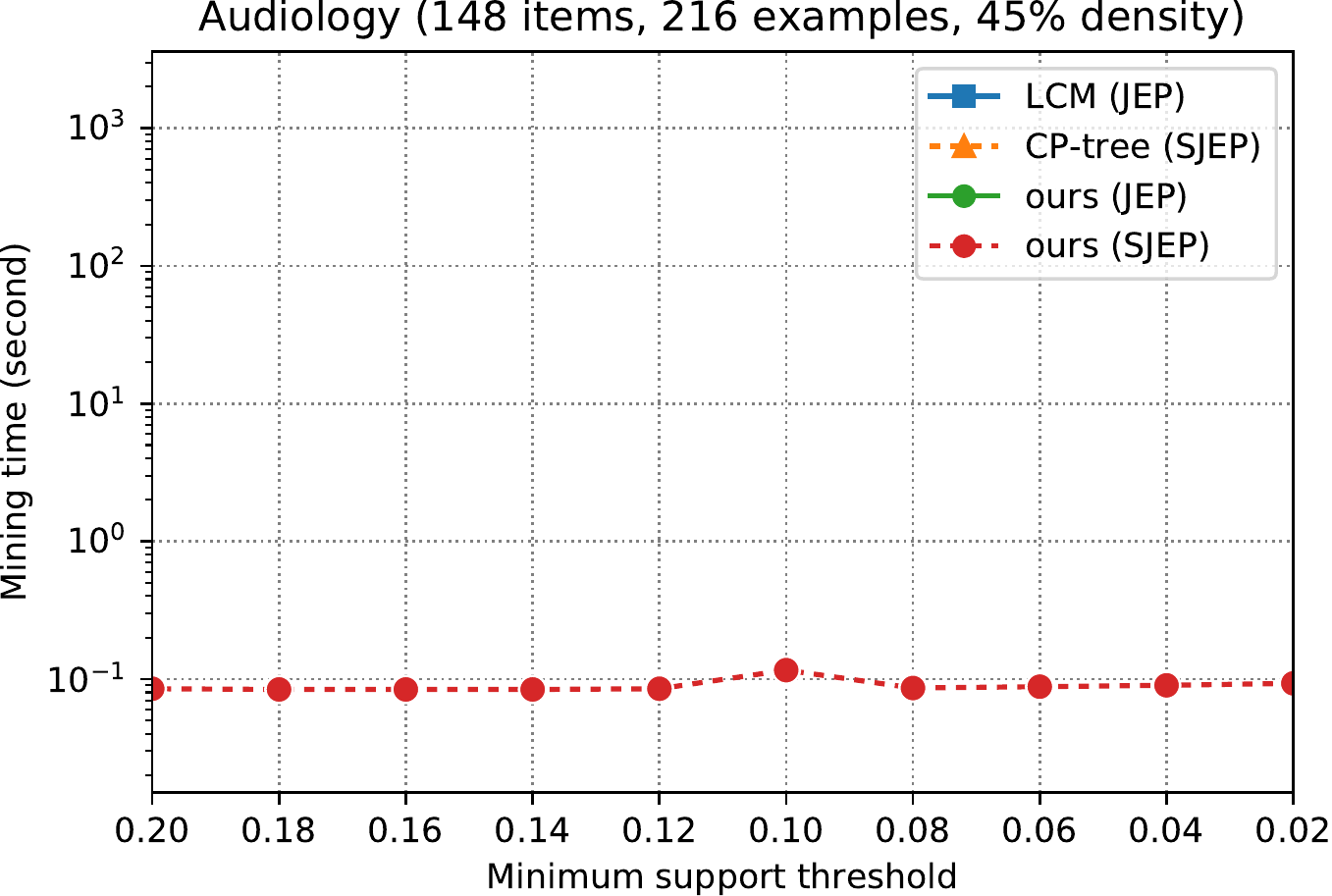}\\
\caption{Comparison of the mining time for JEPs and SJEPs, using LCM, CP-tree, and our algorithm.}
\label{fig:vs-cptree-and-lcm}
\end{figure}

Comparison of mining time for minimal emerging patterns is shown in Figure~\ref{fig:pruning-and-ordering}, where ``LB'' uses the pruning rules 1 and 2, and ``NC'' uses the pruning rules 1 and 3.
Growth rate constraint is fixed to $\theta=9$ in the experiments.
On the Audiology dataset, mining could not be finished within 3600 seconds without the combination of LB and dynamic ordering.
On the dense datasets, effectiveness of LB is improved dramatically when it is combined with dynamic ordering.

\subsection{Experiments 2: Performance comparison with existing methods}

In this experiment, we investigate the speed difference by mining the jumping emerging patterns from the proposed method and the existing methods (LCM~\cite{Uno_FIMI2004} and CP-tree~\cite{Fan_TKDE2006}).
LCM proposed by Uno et al.\ is a state-of-the-art algorithm that won the FIMI 2004 competition with closed frequent itemset mining. 
We used LCM version 5.3\footnote{\url{http://research.nii.ac.jp/~ uno/codes.htm}}, which can mine JEPs by setting large negative weights to the negative data.
The CP-tree proposed by Fan manages the pattern frequency by a tree structure and high-speed mining by reducing access to the database.
We used a C++ implementation of the CP-tree algorithm that mines SJEPs.

The results are shown in Figure~\ref{fig:vs-cptree-and-lcm}.
CP-tree could not complete mining within 3600 seconds on dense datasets.
We can see that LCM can perform JEP mining, which is more expensive than SJEP mining, orders of magnitude faster than the traditional CP-tree algorithm.
On the Audiology dataset, only the SJEP version of our algorithm could be finished within 3600 seconds.
The JEP version of our algorithm sometimes completed orders of magnitudes faster than LCM and the SJEP version was always faster than others in dense datasets.

\subsection{Experiments 3: Evaluation of classification model using mined patterns}

\begin{table}[ht!]
\label{table:f1-score}
\caption{Comparison of F-value between proposed method and existing method.}
\small
\begin{tabularx}{\columnwidth}{cXXXXX}
                       & Proposed method (with negative items) & Proposed method (without negative items) & Decision trees            & Logistic regression      & Random forests      \\ \hline
banknote-class-1        & \textbf{0.998} & 0.992          & 0.989          & 0.991          & 0.994          \\ \hline
breast-tissue-class-adi & \textbf{1.000} & \textbf{1.000} & 0.935          & 0.931          & 0.971          \\
breast-tissue-class-car & \textbf{0.937} & 0.863          & 0.891          & 0.894          & 0.931          \\
breast-tissue-class-con & \textbf{1.000} & \textbf{1.000} & 0.891          & 0.740          & 0.900          \\
breast-tissue-class-fad & \textbf{0.806} & 0.722          & 0.599          & 0.673          & 0.535          \\
breast-tissue-class-gla & \textbf{0.881} & \textbf{0.881} & 0.771          & 0.700          & 0.727          \\
breast-tissue-class-mas & \textbf{0.832} & 0.770          & 0.523          & 0.482          & 0.474          \\ \hline
glass-class-1           & 0.802          & 0.800          & 0.733          & 0.716          & \textbf{0.830} \\
glass-class-2           & \textbf{0.867} & 0.863          & 0.777          & 0.599          & 0.799          \\
glass-class-3           & \textbf{0.697} & 0.625          & 0.558          & 0.231          & 0.371          \\
glass-class-5           & 0.920          & \textbf{0.960} & 0.777          & 0.658          & 0.865          \\
glass-class-6           & \textbf{1.000} & \textbf{1.000} & 0.960          & 0.920          & \textbf{1.000} \\
glass-class-7           & 0.942          & \textbf{0.966} & 0.900          & 0.915          & 0.915          \\ \hline
iris-class-setosa       & \textbf{1.000} & \textbf{1.000} & \textbf{1.000} & \textbf{1.000} & \textbf{1.000} \\
iris-class-versicolor   & \textbf{0.962} & 0.916          & 0.948          & 0.730          & 0.949          \\
iris-class-virginica    & 0.949          & 0.943          & 0.952          & \textbf{0.971} & 0.952          \\ \hline
wifi-class-1            & 0.993          & 0.993          & 0.989          & 0.989          & \textbf{0.997} \\
wifi-class-2            & \textbf{0.982} & 0.961          & 0.979          & 0.977          & 0.978          \\
wifi-class-3            & 0.974          & 0.943          & 0.953          & 0.598          & \textbf{0.975} \\
wifi-class-4            & \textbf{0.995} & 0.956          & 0.991          & 0.994          & \textbf{0.995} \\ \hline
yeast-class-CYT         & 0.632          & 0.604          & 0.604          & 0.606          & \textbf{0.650} \\
yeast-class-ERL.        & \textbf{1.000} & \textbf{1.000} & 0.647          & 0.867          & 0.167          \\
yeast-class-EXC         & \textbf{0.661} & 0.536          & 0.589          & 0.530          & 0.654          \\
yeast-class-ME1         & \textbf{0.785} & 0.734          & 0.761          & 0.641          & 0.779          \\
yeast-class-ME2         & \textbf{0.591} & 0.420          & 0.485          & 0.430          & 0.483          \\
yeast-class-ME3         & \textbf{0.823} & 0.810          & 0.793          & 0.768          & 0.811          \\
yeast-class-MIT         & 0.634          & 0.589          & 0.614          & 0.590          & \textbf{0.645} \\
yeast-class-NUC         & 0.630          & 0.545          & 0.605          & 0.590          & \textbf{0.634} \\
yeast-class-POX         & \textbf{0.628} & 0.614          & 0.614          & 0.614          & 0.560    \\ \hline

\end{tabularx}
\end{table}




In this experiment, we compare the classification model using patterns mined by the proposed method with existing models.
Our model is a generalized additive linear model learned by the LASSO algorithm, which has mined patterns with a maximum length of 5 as features.
We performed binary classification problems on various datasets and compared F-values with existing methods (Logistic Regression, Decision Tree, and Random Forest) by the 5-fold cross validation.
We used binarized data with MDLP for learning of our method, and used original real-valued data for existing methods.
Here we considered two types of binarized data where one introduces negation and the other does not.
Both ours and existing methods tuned hyperparameters with Optuna~\footnote{https:\slash\slash{}optuna.readthedocs.io\slash{}en\slash{}stable\slash}.
We show the used datasets in Table~\ref{tbl:dataset-pred}.

We show the comparison results in Fig~\ref{table:f1-score}.
The results attract attention to that our model shows superior performance in all datasets.
In addition, the following two characteristical results are observed.
(i) Our model achieves high F-values for the datasets causing inferior performances of existing methods, for example the breast-class-mas data and the glass-class-3 data.
(ii) Our model tends to make high performance with introducing negation in the binarized data, for example the yeast-class-EXC data and the yeast-class-ME2 data.



\section{Conclusions}
\label{sec:conc}
In this paper, we consider the problem of constrained pattern mining. We propose dynamic variable-ordering during pattern search, and using dancing links data structures. By computational experiments on real datasets, we observed that our algorithm outperformed the existing algorithms for dense databases.


\bibliographystyle{abbrv}
\bibliography{ref_wl}

\end{document}